\documentclass[10pt,conference]{IEEEtran}
\IEEEoverridecommandlockouts

\usepackage[T1]{fontenc}
\usepackage[dvipsnames]{xcolor}
\usepackage{color}
\definecolor{QCPbeige}{HTML}{E7D5AB}     
\definecolor{QCPblue}{HTML}{479093}      
\definecolor{QCPdarkblue}{HTML}{2C6071}  
\definecolor{QCPgray}{HTML}{333743}      
\definecolor{QCPlightblue}{HTML}{A6CBB9} 
\definecolor{QCPred}{HTML}{E08B8A}       

\usepackage{pgfplots}
\pgfplotsset{compat = newest}
\usepgfplotslibrary{groupplots}
\usepackage{adjustbox}
\usepackage[utf8]{inputenc}
\usepackage{physics}
\usepackage{bm} 
\usepackage{bbm} 
\usepackage{braket}
\usepackage{dsfont}
\usepackage{amsmath}
\usepackage{amssymb}
\usepackage{amsthm}
\usepackage{thmtools} 
\usepackage{mathtools}
\usepackage{tikz}
\usetikzlibrary{
    shapes.geometric,
    backgrounds,
    positioning,
    shapes.geometric,
    decorations.markings,
    decorations.pathreplacing,
    arrows,
    knots,
    hobby,
    angles,
    quotes
}
\makeatletter
\newcommand{\gettikzxy}[3]{%
    \tikz@scan@one@point\pgfutil@firstofone#1\relax
    \edef#2{\the\pgf@x}%
    \edef#3{\the\pgf@y}%
}
\makeatother
\usepackage{graphics}
\usepackage{float}
\usepackage[linesnumbered,ruled]{algorithm2e}
\SetArgSty{textnormal}
\SetAlgorithmName{Protocol}{Protocol}{Lists of protocols}

\usepackage[colorlinks]{hyperref}
\hypersetup{
    colorlinks  = true,
    citecolor   = PineGreen,
    linkcolor   = PineGreen,
    urlcolor    = PineGreen
}

\usepackage{etoolbox}
\apptocmd{\sloppy}{\hbadness 10000\relax}{}{}

\declaretheorem[style=plain]{theorem}

\declaretheorem[style=plain,sibling=theorem]{lemma}
\declaretheorem[style=plain,sibling=theorem]{proposition}

\declaretheorem[numbered=no,name=Assumption]{assumption}

\DeclareMathOperator{\ancilla}{a}     
\DeclareMathOperator{\PSD}{PSD}       
\DeclareMathOperator{\supp}{supp}     
\DeclareMathOperator{\spn}{span}      

\newcommand*{\N}{\mathbb{N}}          
\newcommand*{\R}{\mathbb{R}}          
\newcommand*{\C}{\mathbb{C}}          
\newcommand*{\hil}{\mathcal{H}}       
\newcommand*{\defcolon}{\,:\,}        
\newcommand*{\transpose}{T}           

\newcommand*{\bit}{\mathbbmss{b}}     
\newcommand*{\qubit}{\mathbbmss{q}}   
\newcommand*{\solset}{S}              
\newcommand*{\solspace}{\mathcal{S}}  
\newcommand*{\one}{\mathds{1}}        
\newcommand*{\projection}{\mathds{P}} 
\newcommand*{\lo}{\mathcal{L}}        
\newcommand*{\mixeru}{U_{\text{M}}}   

\newcommand*{\scriptin}{\raisebox{0.15ex}{$\scriptscriptstyle \in$}}
\newcommand*{\deepscriptin}{\raisebox{0.1ex}{$\scriptscriptstyle \in$}}
\newcommand*{\mathhphantomdisplay}[2]{\mathmakebox[\widthof{$\displaystyle #1$}][l]{#2}}

\makeatletter
\newenvironment{taggedsubequations}[1]{
    \addtocounter{equation}{-1}
    \edef\taggedsubeq@anchor{taggedsubequations.\detokenize{#1}}
    \renewcommand{\theHequation}{\taggedsubeq@anchor.parent}
    \begin{subequations}
        \def\@currentlabel{#1}
        \renewcommand{\theequation}{#1.\arabic{equation}}
        \renewcommand{\theHequation}{\taggedsubeq@anchor.\arabic{equation}}
    }
{\end{subequations}}
\makeatother

\makeatletter
\def\ps@IEEEtitlepagestyle{%
  \def\@oddfoot{\mycopyrightnotice}%
  \def\@evenfoot{}%
}
\def\mycopyrightnotice{%
  {\footnotesize
  \hfill 
  \parbox{\textwidth}{%
  © 2026 IEEE.  Personal use of this material is permitted.  Permission from IEEE must be obtained for all other uses, in any current or future media, including reprinting/republishing this material for advertising or promotional purposes, creating new collective works, for resale or redistribution to servers or lists, or reuse of any copyrighted component of this work in other works.}
  \hfill}
}
\makeatother

\begin{document}

\title{
    One for All: A Universal Quantum Conic Programming Framework for Hard-Constrained Combinatorial Optimisation Problems
    \thanks{This work was supported by the DFG through SFB 1227 (DQ-mat), Quantum Frontiers, the Quantum Valley Lower Saxony, the BMBF projects ATIQ and QuBRA.
        The authors used Claude Code for language editing and image generation throughout this article and for distributing the code via GitHub.
        All content was reviewed and edited by the authors, who take full responsibility for the final work.

        \noindent\textbf{Data and code availability statement.}
        The depicted data can be found at \url{https://github.com/Timo59/qce26_numerics}.
    }
}

\author{
    \IEEEauthorblockN{
        Lennart~Binkowski\IEEEauthorrefmark{1},
        Tobias~J.~Osborne\IEEEauthorrefmark{1},
        Marvin~Schwiering\IEEEauthorrefmark{1},
        Ren\'{e}~Schwonnek\IEEEauthorrefmark{1},
        Timo~Ziegler\IEEEauthorrefmark{1}
    }
    \IEEEauthorblockA{
        \IEEEauthorrefmark{1}Institut f\"{u}r Theoretische Physik\\
        Leibniz Universit\"{a}t Hannover, Hannover, Germany\\
        Email: $\{$lennart.binkowski, tobias.osborne, marvin.schwiering, rene.schwonnek, timo.ziegler$\}$@itp.uni-hannover.de
    }
}

\maketitle

\begin{abstract}
    We present a unified quantum-classical framework for addressing NP-complete constrained combinatorial optimisation problems, generalising the recently proposed Quantum Conic Programming (QCP) approach.
    Accordingly, it inherits many favourable properties of the original proposal such as preventing barren plateaus and NP-hard parameter optimisation.
    By collecting the entire classical feasibility structure in a single constraint, we enlarge QCP's scope to arbitrary hard-constrained problems.
    Yet, we prove that the additional restriction is mild enough to still allow for an efficient parameter optimisation via the formulation of a generalised eigenvalue problem (GEP) of adaptable dimension.
    Our rigorous proof further fills some apparent gaps in prior derivations of GEPs from parameter optimisation problems.
    We further detail a measurement protocol for formulating the classical parameter optimisation that does not require us to implement any problem-specific objective Hamiltonian or a quantum feasibility oracle.
    Lastly, we prove that, even under the influence of noise, QCP's parameterised ansatz class always captures the optimum attainable within its generated subcone.
    All of our results hold true for arbitrarily-constrained combinatorial optimisation problems.
\end{abstract}


\section{\label{section:Introduction}Introduction}

Optimisation is a ubiquitous aspect of various disciplines, arising naturally in economics, society, and science.
Decades of advancements in electrical engineering and algorithmic design have pushed the frontier of what is feasible to compute and solve, to the point where even minute improvements often lead to meaningful gains.
One promising candidate for significant advancement is the development of quantum computers, where recent experimental milestones in hardware manufacturing~\cite{PsiquantumTeam2025AManufacturablePlatformForPhotonicQuantumComputing} and error correction~\cite{GoogleQuantumAi2023SuppressingQuantumErrorsByScalingASurfaceCodeLogicalQubit} are expected to ultimately pave the way from the current noisy intermediate-scale quantum (NISQ)~\cite{Preskill2018QuantumComputingInTheNisqEraAndBeyond} devices to fault-tolerant, large-scale quantum processing units (QPUs).
Seminal works by Grover~\cite{Grover1997QuantumMechanicsHelpsInSearchingForANeedleInAHaystack}, Kadowaki and Nishimori~\cite{Kadowaki1998QuantumAnnealingInTheTransverseIsingModel}, and Farhi \textit{et al.}~\cite{Farhi2000QuantumComputationByAdiabaticEvolution} have sparked lasting interest in quantum search and quantum optimisation~\cite{Abbas2024ChallengesAndOpportunitiesInQuantumOptimization}, leading to the proposal of frameworks like quantum branch-and-bound~\cite{Montanaro2020QuantumSpeedupOfBranchAndBoundAlgorithms}, quantum dynamic programming~\cite{Ambainis2019QuantumSpeedupsForExponentialTimeDynamicProgrammingAlgorithms}, and quantum solvers for semidefinite programs~(SDPs)~\cite{VanApeldoorn2020QuantumSdpSolversBetterUpperAndLowerBounds}.
Variational quantum algorithms~(VQAs)~\cite{Cerezo2021VariationalQuantumAlgorithms}---the update of parameters $\theta_{j} \in \R$ controlling a parametrised quantum circuit~(PQC) $V_{L}(\theta_{L}) \cdots V_{1}(\theta_{1})$ by a classical optimisation algorithm---are often deemed to be the most suitable framework in the NISQ-era.
The parameters are optimised in order to yield a high-quality approximation of a given target state, typically the ground state of a target Hamiltonian.
However, significant hurdles remain for the widespread application to classical optimisation problems:
The design of PQCs that respect an optimisation problem's constraints is highly non-trivial and problem-specific~\cite{Hadfield2019FromTheQuantumApproximateOptimizationAlgorithmToAQuantumAlternatingOperatorAnsatz}, the parameter optimisation problem handed to the classical optimiser is often at least as hard as the original problem~\cite{Bittel2021TrainingVariationalQuantumAlgorithmsIsNpHard}, and VQAs often suffer from barren plateaus~\cite{Larocca2025BarrenPlateausInVariationalQuantumComputing}, where vanishing gradients hinder convergence to an even local optimum.

We propose a natural generalisation of Quantum Conic Programming~(QCP), as introduced in~\cite{Binkowski2025FromBarrenPlateausThroughFertileValleysConicExtensionsOfParameterisedQuantumCircuits}, by promoting QCP from a VQA subroutine for unconstrained optimisation problems to a general framework for constrained combinatorial optimisation problems.
QCP and similar frameworks, such as McClean \textit{et al.}'s method for determining excited states~\cite{Mcclean2017HybridQuantumClassicalHierarchyForMitigationOfDecoherenceAndDeterminationOfExcitedStates}, Huggins \textit{et al.}'s non-orthogonal variational quantum eigensolver~\cite{Huggins2020ANonOrthogonalVariationalQuantumEigensolver} and Bharti \textit{et al.}'s NISQ algorithm for SDPs~\cite{Bharti2022NoisyIntermediateScaleQuantumAlgorithmForSemidefiniteProgramming}, evade barren plateaus and avoid the overly complex parameter optimisation by utilising non-unitary parameterised gates.
These gates are implemented via linear combinations of unitaries (LCU)~\cite{Childs2012HamiltonianSimulationUsingLinearCombinationsOfUnitaryOperations} and the classical optimisation involves solving a generalised eigenvalue problem (GEP) or, more generally, an SDP\@.
Our generalisation exploits the fundamental design principles of QCP while extending its scope to constrained problems:
We enforce feasibility preservation of an LCU step by introducing an additional constraint to the associated parameter optimisation problem; yet, under minimal assumptions, we ensure that the resulting optimisation problem can again be reduced to a GEP without incurring additional sampling overhead for the QPU.
Moreover, we obtain a feasibility-preserving LCU operation even if the individual unitaries are not guaranteed to fully preserve feasibility.

After reviewing basics of constrained combinatorial optimisation problems, VQAs---in particular, the Quantum Approximate Optimisation Algorithm~\cite{Farhi2014AQuantumApproximateOptimizationAlgorithm} and the Quantum Alternating Operator Ansatz (QAOA)~\cite{Hadfield2019FromTheQuantumApproximateOptimizationAlgorithmToAQuantumAlternatingOperatorAnsatz}---and QCP in \autoref{section:Preliminaries}, we present our general framework in \autoref{section:Framework}.
We first discuss the incorporation of hard-constraints into QCP's parameter optimisation problem and then formulate the optimisation task via sample matrices.
For the latter, we provide a specialised measurement protocol that utilises the underlying classical structure.
Subsequently, we rigorously reduce the original doubly-constrained parameter optimisation problem to a GEP.
Our detailed proof of \autoref{theorem:GeneralisedEigenvalueProblem} not only extends, but also completes earlier derivations of the singly-constrained case in terms of mathematical rigour and proper handling of degenerate cases.
We further comment on LCU implementation and its success probability, once optimal parameters have been found, and we later discuss our framework's interplay with, and conceptual independence from, QAOA.
We conclude this section by demonstrating that even in the case of general quantum states (density operators), the LCU is readily implemented by a pure state on the ancilla register.
\autoref{section:NumericalExperiments} presents a proof of concept of hard-constrained QCP for small knapsack instances.
In \autoref{section:DiscussionAndConclusion}, we discuss the implications and open questions raised by our work.
We highlight our framework's interplay with QAOA, as well as its conceptual independence from it.

\section{\label{section:Preliminaries}Preliminaries}

\subsection{\label{subsection:ConstrainedCombinatorialOptimisationProblems}Constrained combinatorial optimisation problems}

Let $\bit \coloneqq \{0, 1\}$.
A \emph{constrained combinatorial optimisation problem} (CCOP) is formally defined as a triple $(n, c, \solset)$, where $n \in \N$ is the problem size, $c : \bit^{n} \rightarrow \R$ is the objective function, and $\solset \subseteq \bit^{n}$ is the feasible set.
Given a CCOP, the task is to minimise $c$ over all $\bm{b} \in \solset$.
Prominent examples of CCOPs are the Job-Shop Scheduling problem~\cite{Zhang2017ReviewOfJobShopSchedulingResearchAndItsNewPerspectivesUnderIndustry40}, the Knapsack problem~\cite{Kellerer2004KnapsackProblems}, the Maximum Independent Set problem~\cite{Tarjan1977FindingAMaximumIndependentSet}, and the Travelling Salesperson Problem~\cite{Applegate2009CertificationOfAnOptimalTspTourThrough85900Cities}, for many of which the corresponding decision problem is NP-complete~\cite{Karp1972ReducibilityAmongCombinatorialProblems}.
Throughout this article, we will only consider CCOPs that lie in NP as this implies that, for any given $\bm{b} \in \bit^{n}$, calculating $c(\bm{b})$ and verifying whether $\bm{b} \in \solset$ can be done in time polynomial in $n$.

Let us first quantise the problem:
We represent each classical bit by a qubit $\qubit \cong \C^{2}$, yielding an $n$-qubit register $\hil \coloneqq \qubit^{\otimes n}$.
The classical bit strings constitute the \emph{computational basis} (CB) $\{\ket{\bm{b}} \defcolon \bm{b} \in \bit^{n}\}$ of $\hil$.
The solution space $\solspace$ is the subspace of all superpositions of the feasible CB states, i.e., $\solspace \coloneqq \spn(\ket{\bm{b}} \defcolon \bm{b} \in \solset)$.
Let $C \in \lo(\hil)$ be the objective Hamiltonian, which is diagonal in the CB and whose action on the CB is defined by $C \ket{\bm{b}} \coloneqq c(\bm{b}) \ket{\bm{b}}$.
Then, the minimisation readily translates to finding an eigenstate of $C$ with minimal eigenvalue within the solution space $\solspace$, i.e., to finding a ground state of $C\vert_{\solspace}$.
By the Rayleigh-Ritz inequality, the ground state search, in turn, is equivalent to finding a minimiser of $\braket{\psi | C | \psi}$ among all feasible states $\ket{\psi} \in \solspace$.

VQAs are hybrid heuristics to approximate such a ground state on a quantum computer.
In a nutshell, a VQA introduces a parameterised manifold of attainable states $\ket{\bm{\theta}} = V\!(\bm{\theta}) \ket{\iota}$ by specifying an initial state $\ket{\iota}$ as well as a unitary \emph{parameterised quantum circuit} (PQC) $V\!(\bm{\theta}) \coloneqq V_{L}(\theta_{L}) \cdots V_{2}(\theta_{2}) V_{1}(\theta_{1})$.
The actual parameter optimisation is carried out classically with queries to the quantum computer for estimates of the energy expectation value $E(\bm{\theta}) \coloneqq \braket{\bm{\theta} | C | \bm{\theta}}$ via repeated state preparation and subsequent measurement.
Processing only quantum circuits of low depth, VQAs thereby mitigate the effects of quantum decoherence~\cite{Brandt1999QubitDevicesAndTheIssueOfQuantumDecoherence}.
Most of the optimisation methods used, although sophisticated and varied~\cite{Kuebler2020AnAdaptiveOptimizerForMeasurementFrugalVariationalAlgorithms}, are based on gradient optimisation strategies or (quasi-)Newton methods to perform line search on $E(\bm{\theta})$~\cite{BonetMonroig2023PerformanceComparisonOfOptimizationMethodsOnVariationalQuantumAlgorithms}.
For convex optimisation problems these methods converge to the global optimum.
However, $E(\bm{\theta})$ very rarely is a convex function, while the vast majority of optimisation landscapes are unfavourably non-convex.

\subsection{\label{subsection:QuantumAlternatingOperatorAnsatz}Quantum alternating operator ansatz}

The pioneering VQA for unconstrained combinatorial optimisation problems is Farhi~\textit{et~al.}'s QAOA.
Its PQC is a $p$-fold application of phase separator gates $\exp(-i \gamma C)$ and mixer gates $\exp(-i \beta B)$, where
\begin{align}\label{equation:QAOAMixerHamiltonian}
    B \coloneqq \sum_{i = 1}^{n} \sigma_{x}^{(i)}
\end{align}
is the QAOA-mixer Hamiltonian, to the initial state $\ket{+}$, the uniform superposition of all CB states.
Initial state, mixer, and phase separator are chosen to approximate the ground state arbitrarily well for sufficiently deep circuits~\cite{Binkowski2024ElementaryProofOfQaoaConvergence}.
The QAOA's scope can be extended to general CCOPs by introducing soft-constraints:
The objective function $c$ now penalises infeasible bit strings, e.g., by choosing a penalty $\alpha > 0$ and redefining
\begin{align}\label{equation:QAOASoftConstrained}
    \tilde{c}(\bm{b}) \coloneqq \begin{cases}
        c(\bm{b}),& \text{if } \bm{b} \in \solset \\
        c(\bm{b}) + \alpha,& \text{else}.
    \end{cases}
\end{align}
If $\alpha$ is sufficiently large, the minima of $\tilde{c}$ and $c$ coincide, so applying QAOA to~\eqref{equation:QAOASoftConstrained} is expected to yield a good approximation of the optimal solution.
However, empirical studies show that the hyperparameter $\alpha$ has to be carefully fine-tuned to sufficiently suppress infeasibility on the one hand, but to not dominate the optimisation landscape on the other~\cite{GrandRive2019KnapsackProblemVariantsOfQaoaForBatteryRevenueOptimisation}.

As an alternative to using soft-constraints, Hadfield~\textit{et~al.} developed a more general framework under the same acronym QAOA, which considers feasible initial states and feasibility-preserving mixers to force constraint satisfaction throughout the entire routine, rendering additional soft-constraints obsolete.
Formally, an $n$-qubit operator $A \in \lo(\hil)$ is called feasibility-preserving if it maps feasible states to feasible states, i.e., if $A(\solspace) \subseteq \solspace$.
Provided the initial state $\ket{\iota} \in \solspace$ is feasible, the framework effectively restricts the entire problem to the subspace $\solspace$.
Note that since the objective Hamiltonian $C$ is diagonal in the CB, so is the phase separator $\exp(-i \gamma C)$ for all parameter values $\gamma$;
hence it is especially feasibility-preserving.
Restricting the optimisation to the feasible subspace $\solspace$ therefore ultimately depends on the QAOA-mixer.

A proper QAOA-mixer $\mixeru$ not only preserves feasibility, but also fulfils certain mixing properties.
In~\cite{Hadfield2019FromTheQuantumApproximateOptimizationAlgorithmToAQuantumAlternatingOperatorAnsatz}, $\mixeru$ is required to provide transitions between all feasible CB states, that is for every pair $\bm{b}, \bm{y} \in \solset$ there exists some parameter value $\beta^{*} \in \R$ and some power $r \in \N$ such that $\braket{\bm{b} | \mixeru^{r}(\beta^{*}) | \bm{y}} \neq 0$.
Strengthening this condition to irreducibility and element-wise non-negativity of the reduced mixer Hamiltonian $B_{\text{M}}\vert_{\solspace}$, reintroducing the two crucial properties of \eqref{equation:QAOAMixerHamiltonian}, such that $\mixeru(\beta) = \exp(-i \beta B_{\text{M}})$ allows to deduce the same convergence result as for the original QAOA~\cite{Binkowski2024ElementaryProofOfQaoaConvergence}.

Hadfield~\textit{et~al.}'s generalisation of the QAOA has undoubtedly advanced the variational toolbox for CCOPs and has influenced numerous problem-specific improvements of the QAOA.
However, it is lacking the charm of the original QAOA proposal regarding out-of-the-box applicability:
In the original proposal, initial state and mixer are problem-agnostic.
The problem's structure solely enters via the phase separator and the expectation value, and both are uniquely determined by the classical objective function (modulo soft-constraints).
In contrast, the construction of a problem-specific mixer does not follow a universal pattern.
Mixers need to be handcrafted individually for each CCOP and typically require deep insights into the problem's feasibility structure and symmetries.
For example,~\cite{Hadfield2019FromTheQuantumApproximateOptimizationAlgorithmToAQuantumAlternatingOperatorAnsatz} already gives constructions for several important CCOPs such as the Travelling Salesperson Problem~(TSP) which do not follow a clear unique engineering rule.

The proposal of Grover-mixer QAOA~\cite{Bartschi2020GroverMixersForQaoaShiftingComplexityFromMixerDesignToStatePreparation} gives a somewhat unified approach to constructing feasibility-preserving mixers:
For a given CCOP, let $\ket{\solset} \coloneqq \lvert \solset\rvert^{-1 / 2} \sum_{\bm{b} \in \solset} \ket{\bm{b}}$ be the uniform superposition of all feasible CB states and $U_{\solset} \in \lo(\hil)$ be the unitary circuit generating that superposition, i.e., $U_{\solset} \ket{\bm{0}} = \ket{\solset}$.
Then, $\exp(- i \beta \ketbra{\solset}{\solset})$ constitutes a valid QAOA-mixer in the sense of~\cite{Hadfield2019FromTheQuantumApproximateOptimizationAlgorithmToAQuantumAlternatingOperatorAnsatz}.
Note that this mixer can be constructed from $U_{\solset}$ via $\exp(-i \beta \ketbra{\solset}{\solset}) = U_{\solset}^{\vphantom{\dagger}} (\one - (1 - e^{-i \beta}) \ketbra{\bm{0}}{\bm{0}}) U_{\solset}^{\dagger}$.
Formally, this ansatz introduces a unique mixer design and specific constructions for Max $k$-VertexCover, the TSP, and discrete portfolio re-balancing were discussed in~\cite{Bartschi2020GroverMixersForQaoaShiftingComplexityFromMixerDesignToStatePreparation}.
Crucially, there is no general recipe how to construct the essential state preparation circuit $U_{\solset}$ for a given problem.

Lastly, all discussed variants of QAOA, like most VQAs share the tendency of exhibiting \emph{barren plateaus} (BPs)~\cite{Mcclean2018BarrenPlateausInQuantumNeuralNetworkTrainingLandscapes}.
For PQCs that can produce all states sufficiently uniformly, the gradient of $E(\bm{\theta})$ has vanishing mean and variance decreasing exponentially in the number of qubits.
Hence starting from random initial parameters, each parameter update requires a number of measurements scaling exponentially in the number of qubits.
The parameter landscape becomes practically flat almost everywhere which affects gradient-based methods and gradient-free optimisation routines alike~\cite{Arrasmith2021EffectOfBarrenPlateausOnGradientFreeOptimization}.
BPs also emerge from quantum device noise for PQCs growing linearly with the number of qubits~\cite{Wang2021NoiseInducedBarrenPlateausInVariationalQuantumAlgorithms}, from quantum entanglement itself~\cite{OrtizMarrero2021EntanglementInducedBarrenPlateaus}, and even exist in shallow circuits, provided that $E(\bm{\theta})$ is the expectation value of global observables~\cite{Cerezo2021CostFunctionDependentBarrenPlateausInShallowParametrizedQuantumCircuits}.

\subsection{\label{subsection:QuantumConicProgramming}Quantum conic programming}

Across all versions of the QAOA, the parameter updates carried out by the CPU are multi-dimensional, non-convex optimisation problems.
These are notoriously hard, in a certain sense even harder than the ground state search they are derived from~\cite{Bittel2021TrainingVariationalQuantumAlgorithmsIsNpHard}.
This is why some recent proposals, such as Huggings~\textit{et~al.}'s non-orthogonal variational quantum eigensolver~\cite{Huggins2020ANonOrthogonalVariationalQuantumEigensolver}, suggest VQAs whose parameter optimisation admits \emph{convex} objective functions for the classical computer.
The proposed algorithms share the formulation of the parameter optimisation as a \emph{generalised eigenvalue problem} (GEP) or, more generally, as a \emph{semidefinite program} (SDP), and the implementation of the quantum state update via \emph{Linear Combinations of Unitaries} (LCU)~\cite{Childs2012HamiltonianSimulationUsingLinearCombinationsOfUnitaryOperations}.
In the following, we provide a more detailed recap of QCP as introduced in~\cite{Binkowski2025FromBarrenPlateausThroughFertileValleysConicExtensionsOfParameterisedQuantumCircuits} which is already specifically tailored to (unconstrained) combinatorial optimisation problems and, similar to other proposals, designed in such a way that the effect of BPs are mitigated.

Starting from some state $\ket{\phi}$ which is the result of applying an already optimised QAOA-PQC to the initial state $\ket{+}$, we introduce the set of unitaries $\{U_{i}\}_{i = 1}^{3}$ with $U_{1} = \exp(-i \delta_{1} B)$, $U_{2} = \exp(-i \delta_{2} C)$, and $U_{3} = \one$ with fixed (random) angles $\delta_{1}$ and $\delta_{2}$.\footnote{
    The construction is actually formulated for an arbitrary number $\ell \in \N$ of unitaries, but we restrict this recap to the QAOA-specific case $\ell = 3$.
}
Using these unitaries, we define the parameterised ansatz class
\begin{align}\label{equation:AnsatzClass}
    \mathcal{M}_{\bm{\alpha}} \ket{\phi} \coloneqq \frac{M_{\bm{\alpha}} \ket{\phi}}{\norm{M_{\bm{\alpha}} \ket{\phi}}}\ \text{with}\ M_{\bm{\alpha}} \coloneqq \sum_{j = 1}^{3} \alpha_{j} U_{j},\ \bm{\alpha} \in \C^{3}.
\end{align}
The $\mathcal{M}_{\bm{\alpha}}$ are therefore non-unitary LCU-circuits, parameterised by three complex numbers $\bm{\alpha}$.
Note that since $U_{3} = \one$, choosing $\bm{\alpha} = (0, 0, 1)$ always reproduces the initial state $\ket{\phi}$;
hence minimising over the ansatz class \eqref{equation:AnsatzClass} cannot worsen the result.
The parameter optimisation for the proposed ansatz class translates into
\begin{taggedsubequations}{QCP-OPT}\label{equation:QCPParameterOptimisation}
    \begin{alignat}{2}
        &\min_{\bm{\alpha} \in \C^{3}} &&\braket{\phi | M_{\bm{\alpha}}^{\dagger} C M_{\bm{\alpha}}^{\vphantom{\dagger}} | \phi} \\
        &\text{ s.t. } &&\braket{\phi | M_{\bm{\alpha}}^{\dagger} M_{\bm{\alpha}}^{\vphantom{\dagger}} | \phi} = 1. \label{equation:QCPParameterOptimisation2}
    \end{alignat}
\end{taggedsubequations}

In order to obtain a tangible version of the abstract problem \eqref{equation:QCPParameterOptimisation} to be handed to a CPU, we introduce the sample matrices $\mathbf{E}, \mathbf{H} \in \C^{3 \times 3}$ defined via
\begin{align*}
    \mathbf{E}_{j k} \coloneqq \braket{\phi | U_{j}^{\dagger} U_{k}^{\vphantom{\dagger}} | \phi}\ \text{and}\ \mathbf{H}_{j k} \coloneqq \braket{\phi | U_{j}^{\dagger} C U_{k}^{\vphantom{\dagger}} | \phi}
\end{align*}
so that we can compactly write
\begin{align}
    \braket{\phi | M_{\bm{\alpha}}^{\dagger} M_{\bm{\alpha}}^{\vphantom{\dagger}} | \phi} = \sum_{j, k = 1}^{3} \overline{\alpha_{j}} \braket{\phi | U_{j}^{\dagger} U_{k}^{\vphantom{\dagger}} | \phi} \alpha_{k} = \bm{\alpha}^{\dagger} \mathbf{E} \bm{\alpha}, \label{equation:SampleMatricesIdentity1}\ \text{and}\\
    \braket{\phi | M_{\bm{\alpha}}^{\dagger} C M_{\bm{\alpha}}^{\vphantom{\dagger}} | \phi} = \sum_{j, k = 1}^{3} \overline{\alpha_{j}} \braket{\phi | U_{j}^{\dagger} C U_{k}^{\vphantom{\dagger}} | \phi} \alpha_{k} = \bm{\alpha}^{\dagger} \mathbf{H} \bm{\alpha}. \label{equation:SampleMatricesIdentity2}
\end{align}
Therefore, \eqref{equation:QCPParameterOptimisation} has been transformed into
\begin{align*}
    &\min_{\bm{\alpha} \in \C^{3}} \bm{\alpha}^{\dagger} \mathbf{H} \bm{\alpha} \\
    &\text{s.t. } \bm{\alpha}^{\dagger} \mathbf{E} \bm{\alpha} = 1
\end{align*}
which may, in turn, be reformulated as a three-dimensional GEP $\mathbf{H} \bm{\alpha} = \lambda \mathbf{E} \bm{\alpha}$, where the minimal generalised eigenvalue of $\mathbf{H}$ and $\mathbf{E}$ and the corresponding generalised eigenvector coincide with the optimum and optimiser of \eqref{equation:QCPParameterOptimisation}.
We defer the discussion on how to obtain the matrix elements of $\mathbf{E}$ and $\mathbf{H}$ as well as the detailed problem derivations to the \hyperref[section:Framework]{next section}.

Lastly, we wish to implement the LCU with parameters corresponding to the optimiser of \eqref{equation:QCPParameterOptimisation}.
This generally non-unitary and thus inherently non-deterministic step comes with a certain failure probability which can be suppressed, although never eliminated, by more involved state preparation and measurement protocols on a two-qubit ancilla register.
The best strategy found in~\cite{Binkowski2025FromBarrenPlateausThroughFertileValleysConicExtensionsOfParameterisedQuantumCircuits} is to prepare the state
\begin{align*}
    \ket{\psi}_{\ancilla} \coloneqq \frac{(\sqrt{\bm{\alpha}}, 0)}{\norm{\sqrt{\bm{\alpha}}}_{2}} = \frac{(\sqrt{\bm{\alpha}}, 0)}{\sqrt{\norm{\bm{\alpha}}_{1}}}
\end{align*}
in the ancilla register, where $\sqrt{\bm{\alpha}}$ is the vector obtained by taking the entry-wise (principle) square root of $\bm{\alpha}$, and $(\sqrt{\bm{\alpha}}, 0)$ means to pad this vector with a zero in order to obtain a vector of length $4 = 2^{2}$.
Subsequently, we apply the compound unitary
\begin{align*}
    \mathcal{U} &\coloneqq C_{01}(U_{2}) C_{00}(U_{1}) \\
    &= U_{1} \hspace*{-2pt} \otimes \hspace*{-2pt} \ketbra{00}{00} + U_{2} \hspace*{-2pt} \otimes \hspace*{-2pt} \ketbra{01}{01} + U_{3} \hspace*{-2pt} \otimes \hspace*{-2pt} (\ketbra{10}{10} + \ketbra{11}{11})
\end{align*}
to the product state $\ket{\phi} \otimes \ket{\psi}_{\ancilla}$, where, e.g., $C_{00}(U_{1})$ denotes the unitary $U_{1}$ applied to the main register, controlled on the $\ket{00}$-state in the ancilla register.
Finally, we apply the inverse state preparation circuit of the state
\begin{align*}
    \ket{\xi}_{\ancilla} \coloneqq \frac{\big(\overline{\sqrt{\bm{\alpha}}}, 0\big)}{\norm{\sqrt{\bm{\alpha}}}_{2}} = \frac{\big(\overline{\sqrt{\bm{\alpha}}}, 0\big)}{\sqrt{\norm{\bm{\alpha}}_{1}}}
\end{align*}
to the ancilla register and conduct a binary measurement in the CB, distinguishing between the all-zero bit string and all other bit strings.
The success probability corresponds exactly to the probability of measuring all-zero and is precisely given by $\norm{\bm{\alpha}}_{1}^{-2}$.
Classical post-processing on whether the measurement outcome is all-zero then implements the entire LCU which updates the quantum state according to the chosen parameter vector $\bm{\alpha}$.

\section{\label{section:Framework}Framework}

As a first step, we conceptually decouple QCP from the QAOA.
This allows for formulating the former in a more compact way, highlighting its essential building blocks and their properties.
From this cleansed point of view, we generalise QCP to general CCOPs without introducing soft-constraints, lifting its many favourable features such as mapping the ground state search to an efficiently solvable classical problem as well as evading BPs to a more general class of problems.
Unlike the generalisation of the QAOA discussed in \autoref{section:Preliminaries}, our generalisation of QCP is again constructive in the sense that knowledge of the classical objective function and some classical feasibility oracle suffices to directly construct the entire QCP routine.
Finally, we reconnect the generalised QCP to the QAOA and discuss how problem-specific mixers can further improve the execution of QCP.

\subsection{\label{subsection:ParameterisedAnsatzClass}Parameterised ansatz class}

Let $(n, c, \solset)$ be a generic CCOP in NP.
We call a set of unitaries $\{U_{j}\}_{j = 1}^{\ell} \subset \lo(\hil)$ such that $U_{\ell} = \one$, a collection of search unitaries.
The initial state $\ket{\iota}$ passes through a sequence of $L$ LCUs,
\begin{align}
    \mathcal{M}_{\bm{\alpha}_{a}}^{(a)} \ket{\phi} \coloneqq \frac{\sum_{j_{a} = 1}^{\ell_{a}} \alpha_{j_{a}}^{(a)}\, U_{j_{a}}^{(a)} \ket{\phi}}{\norm{\sum_{j_{a} = 1}^{\ell_{a}} \alpha_{j_{a}}^{(a)}\, U_{j_{a}}^{(a)} \ket{\phi}}}
\end{align}each characterised by a set of search unitaries $\{U^{(a)}_{j_{a}}\}_{j_{a} = 1}^{\ell_{a}}$ and the complex-valued coefficients $\bm{\alpha^{(a)}} \equiv (\alpha_{1}^{(a)}, \dots, \alpha_{\ell_{a}}^{(a)})$.
Upon measurement, given all LCUs have been implemented successfully, the main register is in the state
\begin{align}\label{equation:QCPstate}
    \ket{\phi} = \sum_{j_{1} = 1}^{\ell_{1}} \cdots \sum_{j_{L} = 1}^{\ell_{L}} \alpha_{j_{L}}^{(L)} \cdots \alpha_{j_{1}}^{(1)}\; U_{j_{L}}^{(L)} \cdots U_{j_{1}}^{(1)} \ket{\iota}.
\end{align}
Each optimisation step updates one coefficient vector $\bm{\alpha}^{(k)}$ while all others are held fixed; the schedule assigning $k$ to a step is part of the ansatz.
If the initial state $\ket{\iota}$ is feasible, e.g, it corresponds to the solution of some classical approximation or, more generally, it is the superposition of feasible CB states, the search unitaries can be chosen without any additional requirements.
Otherwise, they are required to allow for initial parameters $(\bm{\alpha}^{(1)}, \dots, \bm{\alpha}^{(L)})$ such that the output state~\eqref{equation:QCPstate} before the optimisation process is feasible, i.e., $\ket{\phi} \in \solspace$.
For example, let $\ket{\iota} \equiv \ket{+}$, $L \equiv n$, and $\{Y_{1}, Y_{2}, \dots, Y_{n}, \one\}$, where $Y_{j}$ is the Pauli-$Y$ gate on the $j$th qubit, be the set of search unitaries for each LCU.
Choosing $\alpha_{a}^{(a)} = (-1)^{b_{a}} \cdot i/\sqrt{2}$ and $\alpha_{n+1}^{(a)} = 1/\sqrt{2}$, with all remaining coefficients of $\bm{\alpha}^{(a)}$ set to zero, gives $M_{\bm{\alpha}^{(a)}}^{(a)} = \left(\one + i (-1)^{b_{a}} Y_{a} \right) / \sqrt{2}$ which maps $\ket{+}$ on the $a$th qubit to $\ket{b_{a}}$.
Applying all $n$ LCUs in sequence thus prepares an arbitrary CB state $\ket{\bm{b}} \in \{\ket{\bm{b}} \;:\; \bm{b} \in \bit^{n} \}$ from $\ket{+}$.

Instead of introducing penalties to the objective Hamiltonian, as described in~\autoref{section:Preliminaries}, we propose to directly introduce hard constraints to the optimisation task:
Let $\projection_{\solspace}$ denote the orthogonal projection onto the feasible subspace $\solspace$, then $\one - \projection_{\solspace}$ is the orthogonal projection onto its orthogonal complement $\solspace^{\perp}$.
In addition to the normalisation constraint \eqref{equation:QCPParameterOptimisation2}, we also require that $(\one - \projection_{\solspace}) \mathcal{M}_{\bm{\alpha}} \ket{\iota} = 0$, i.e.\ that parameterised candidate states are orthogonal to $\solspace^{\perp}$.
In summary, optimising the coefficients $\bm{\alpha} \equiv \bm{\alpha}^{(k)}$, we arrive at the generalised optimisation problem
\begin{taggedsubequations}{GenQCP-OPT}\label{equation:GenQCPParameterOptimisation}
    \begin{alignat}{2}
        &\min_{\bm{\alpha} \in \C^{\ell}} &&\braket{\psi | M_{\bm{\alpha}}^{\dagger}\, T(C) M_{\bm{\alpha}}^{\vphantom{\dagger}} | \psi} \\
        &\text{ s.t. } &&\braket{\psi | M_{\bm{\alpha}}^{\dagger}\, T(\one) M_{\bm{\alpha}}^{\vphantom{\dagger}} | \psi} = 1 \\
        &\hphantom{\text{ s.t. }} &&\braket{\psi | M_{\bm{\alpha}}^{\dagger}\, T(\one - \projection_{\solspace}) M_{\bm{\alpha}}^{\vphantom{\dagger}} | \psi} = 0,
    \end{alignat}
\end{taggedsubequations}
with the output of the sequence prior to the $k$-th LCU $\ket{\psi}$ and $T(A) = (\mathcal{M}_{\bm{\alpha}_{k+1}}^{(k+1)})^{\dagger} \cdots (\mathcal{M}_{\bm{\alpha}_{L}}^{(L)})^{\dagger} A\; \mathcal{M}_{\bm{\alpha}_{L}}^{(L)} \cdots \mathcal{M}_{\bm{\alpha}_{k+1}}^{(k+1)}$.

\subsection{\label{subsection:SampleMatrices}Sample matrices}

The two original QCP sample matrices encode the correct normalisation ($\mathbf{E}$) and the objective ($\mathbf{H}$), but do not contain any information regarding feasibility.
Since we do not modify the Hamiltonian, the $\mathbf{H}$-matrix stays the same.
Renaming $\mathbf{E}$ to $\mathbf{F}$ and introducing an additional matrix $\mathbf{G}$ which quantifies orthogonality to $\solspace^{\perp}$, gives rise to the augmented parameter optimisation.
We start by formalising the process of constructing general sample matrices.
\begin{lemma}\label{lemma:ConstructingSampleMatricesIsCompletelyPositive}
    For every state $\ket{\psi} \in \hil$ and every collection of unitary operators $\{U_{j}\}_{j = 1}^{\ell} \subset \lo(\hil)$, the map
    \begin{align}
        \begin{split}
            T : \lo(\hil) &\rightarrow \C^{\ell \times \ell}, \\
            A &\mapsto \big(\hspace*{-3pt}\braket{\psi | U_{j}^{\dagger} A U_{k}^{\vphantom{\dagger}} | \psi}\hspace*{-3pt}\big)_{j, k = 1}^{\ell}
        \end{split}
    \end{align}
    is positive.
    If the family of states $\{U_{j} \ket{\psi} \defcolon j = 1, \ldots, \ell\}$ is linearly independent and $B \in \lo(\hil)$ is positive-definite, so is $T(B)$ in $\C^{\ell \times \ell}$.
\end{lemma}

\begin{proof}
    If $A \succeq 0$, it can be expressed as $A = X^{\dagger} X$.
    Since
    \begin{align*}
        T(A)_{j k} = \braket{\psi | U_{j}^{\dagger} A U_{k}^{\vphantom{\dagger}} | \psi} = \braket{\psi | (X U_{j}^{\vphantom{\dagger}})^{\dagger} X U_{k}^{\vphantom{\dagger}} | \psi},
    \end{align*}
    $T(A)$ is the Gram matrix of the vectors $\ket{x_{j}} \coloneqq X U_{j} \ket{\psi}$, hence it is positive semidefinite.
    Furthermore,
    \begin{align*}
        0 = v^{\dagger}\, T(A)\, v = \norm*{\sum_{j} v_{j} x_{j}}^{2} = \norm*{X \sum_{j} v_{j}\, U_{j} \ket{\psi}}^{2}
    \end{align*}
    for any $v \ne \bm{0}$ implies that $\sum_{j} v_{j}\, U_{j} \ket{\psi} = \bm{0}$, contradicting linear independence of $\{U_{j} \ket{\psi} \defcolon j = 1, \dots, \ell\}$.
\end{proof}

Our three sample matrices therefore are the images under $T$ of the positive operators $\one$, $\one - \projection_{\solspace}$, and $C$, respectively.%
\footnote{%
    For the sake of readability, we consider QCP with a single LCU until further notice.
    Without further ado, all implications follow for \autoref{equation:GenQCPParameterOptimisation} from~\autoref{lemma:ConstructingSampleMatricesIsCompletelyPositive} as well.
}
That is, we define
\begin{alignat*}{3}
    &\mathbf{F}_{j k} &&\coloneqq T(\one)_{j k} &&= \braket{\iota | U_{j}^{\dagger} \one U_{k}^{\vphantom{\dagger}} | \iota}, \\
    &\mathbf{G}_{j k} &&\coloneqq T(\one - \projection_{\solspace})_{j k} &&= \braket{\iota | U_{j}^{\dagger} (\one - \projection_{\solspace}) U_{k}^{\vphantom{\dagger}} | \iota}, \\
    &\mathbf{H}_{j k} &&\coloneqq T(C)_{j k} &&= \braket{\iota | U_{j}^{\dagger} C U_{k}^{\vphantom{\dagger}} | \iota}.
\end{alignat*}
Since $T$ is positive according to \autoref{lemma:ConstructingSampleMatricesIsCompletelyPositive}, $\mathbf{F}$, $\mathbf{G}$, and $\mathbf{H}$ are positive-semidefinite matrices.
In addition, choosing the unitaries $U_{i}$ such that $\{U_{i} \ket{\iota} \defcolon j = 1, \ldots, \ell\}$ is linearly independent, \autoref{lemma:ConstructingSampleMatricesIsCompletelyPositive} guarantees that $\mathbf{F}$ is positive-definite.
W.l.o.g., we may assume that $C$ is positive-definite (this can be simply achieved by adding a sufficiently large offset to the classical objective function).
Then, again by \autoref{lemma:ConstructingSampleMatricesIsCompletelyPositive}, $\mathbf{H}$ is positive-definite.
For $\mathbf{G}$, we cannot make such a claim in full generality.
A calculation analogous to \eqref{equation:SampleMatricesIdentity1} and \eqref{equation:SampleMatricesIdentity2} shows that \eqref{equation:GenQCPParameterOptimisation} is equivalent to the more tangible problem
\begin{align}\label{equation:GenQCPMatrixParameterOptimisation}
    \begin{split}
        &\min_{\bm{\alpha} \in \C^{\ell}} \bm{\alpha}^{\dagger} \mathbf{H} \bm{\alpha} \\
        &\mathhphantomdisplay{\min_{\bm{\alpha} \in \C^{\ell}}}{\ \,\text{s.t.}}\ \mathhphantomdisplay{\bm{\alpha}^{\dagger} \mathbf{G} \bm{\alpha}}{\bm{\alpha}^{\dagger} \mathbf{F} \bm{\alpha}} = 1 \\
        &\hphantom{\min_{\bm{\alpha} \in \C^{\ell}}}\ \bm{\alpha}^{\dagger} \mathbf{G} \bm{\alpha} = 0.
    \end{split}
\end{align}

Typically, the sample matrices' off-diagonal entries are calculated with the aid of the Hadamard test~\cite{Ekert2002DirectEstimationsOfLinearAndNonlinearFunctionalsOfAQuantumState} and Pauli measurements~\cite{Bharti2021IterativeQuantumAssistedEigensolver} for all Pauli terms in the respective observables.
However, note that we do not have direct access to the infeasible projection $\one - \projection_{\solspace}$.
Additionally, all three operators $\one$, $\one - \projection_{\solspace}$, and $C$ mutually commute as they are all diagonal in the CB, meaning that they can be simultaneously measured.
In fact, we propose to implement neither of these operators directly on the QPU, but to approximate the matrix entries solely via sampling bit strings from the states $U_{j} \ket{\iota}$ via measurements in the CB, and to apply the classical objective function and feasibility oracle, respectively.

The diagonal $\braket{\iota | U_{j}^{\dagger} A U_{j}^{\vphantom{\dagger}} | \iota} = \braket{j | A | j}$, $A \in \{\one, \one - \projection_{\solspace}, C\}$ holds expectation value of $A$ w.r.t. the states $\ket{j} \coloneqq U_{j} \ket{\iota}$.
We jointly estimate them by sampling $m \in \N$ bit strings from $\ket{j}$ and calculating the empirical average of the classical objective function (for $A = C$) and the frequency of infeasible states (for $A = \one - \projection_{\solspace}$), respectively.
For the off-diagonal entries, the polarisation identity implies that
\begin{align}\label{equation:PolarisationIdentity}
    \begin{split}
        \braket{\iota | U_{j}^{\dagger} A U_{k}^{\vphantom{\dagger}} | \iota} &= \frac{1}{2} \big(\braket{\iota | (U_{j}^{\dagger} + U_{k}^{\dagger}) A (U_{j}^{\vphantom{\dagger}} + U_{k}^{\vphantom{\dagger}}) | \iota} \\
        &\quad\quad - \braket{\iota | U_{j}^{\dagger} A U_{j}^{\vphantom{\dagger}} | \iota} - \braket{\iota | U_{k}^{\dagger} A U_{k}^{\vphantom{\dagger}} | \iota}\big) \\
        &\quad - \frac{i}{2} \big(\braket{\iota | (U_{j}^{\dagger} - i U_{k}^{\dagger}) A (U_{j}^{\vphantom{\dagger}} + i U_{k}^{\vphantom{\dagger}}) | \iota} \\
        &\quad\quad - \braket{\iota | U_{j}^{\dagger} A U_{j}^{\vphantom{\dagger}} | \iota} - \braket{\iota | U_{k}^{\dagger} A U_{k}^{\vphantom{\dagger}} | \iota}\big).
    \end{split}
\end{align}
Apart from the $j$th and $k$th diagonal entry,~\eqref{equation:PolarisationIdentity} requires to estimate expectation values of $A$ w.r.t the unnormalised vectors $(U_{j} + U_{k}) \ket{\iota}$ and $(U_{j} + i U_{k}) \ket{\iota}$, for all combinations of $j, k = 1, \ldots, \ell$ such that $j \neq k$.
By hermiticity of the sample matrices, it suffices to restrict the calculation to $k < j$.
Introducing an additional ancilla qubit and the binary LCU operation
\begin{align}\label{equation:SampleMatricesRealPartCircuit}
    \Lambda_{\Re}^{(j k)} &\coloneqq (\one \otimes H) C_{1}(U_{k}) C_{0}(U_{j}) (\one \otimes H),
\end{align}
maps the initial state $\ket{\iota} \otimes \ket{0}$ to
\begin{align*}
    (\one \otimes H) \big(U_{j} \otimes \ketbra{0}{0} + U_{k} \otimes \ketbra{1}{1}\big) (\one \otimes H) \ket{\iota} \otimes \ket{0} \\
    = \frac{1}{2} \big(U_{j} + U_{k}\big) \ket{\iota} \otimes \ket{0} + \frac{1}{2} \big(U_{j} - U_{k}\big) \ket{\iota} \otimes \ket{1}.
\end{align*}
Indeed, this is very similar to the Hadamard test for calculating (the real part of) inner products of states.
A subsequent measurement of the ancilla qubit in the CB with outcome zero readily implements the normalised version of $(U_{j} + U_{k}) \ket{\iota}$ on the main register.
The normalisation factor's inverse is given by
\begin{align}\label{equation:NormalisationFactorReal}
    \big\lVert(U_{j}^{\vphantom{\dagger}} + U_{k}^{\vphantom{\dagger}}) \ket{\iota}\big\rVert &= \sqrt{\braket{\iota | (U_{j}^{\dagger} + U_{k}^{\dagger}) (U_{j}^{\vphantom{\dagger}} + U_{k}^{\vphantom{\dagger}}) | \iota}} \\
    &= \sqrt{2 + 2 \Re(\braket{\iota | U_{j}^{\dagger} U_{k}^{\vphantom{\dagger}} | \iota})}.
\end{align}
Similarly, the probability of measuring the ancilla qubit in $\ket{0}$ is
\begin{align}\label{equation:ProbabilityReal}
    p_{j k}(0) = \frac{1}{4} \big\lVert(U_{j}^{\vphantom{\dagger}}\hspace*{-2pt} +\hspace*{-1pt} U_{k}^{\vphantom{\dagger}})\hspace*{-1pt} \ket{\iota}\hspace*{-2pt}\big\rVert^{2} = \frac{1}{2} \big(1\hspace*{-1pt} +\hspace*{-1pt} \Re(\braket{\iota | U_{j}^{\dagger} U_{k}^{\vphantom{\dagger}} | \iota})\big).
\end{align}
Interleaving~\eqref{equation:SampleMatricesRealPartCircuit} with an additional phase gate yields
\begin{align*}
    \Lambda_{\Im}^{(j k)} &\coloneqq (\one \otimes H) C_{1}(U_{k}) C_{0}(U_{j}) (\one \otimes S) (\one \otimes H),
\end{align*}
which maps $\ket{\iota} \otimes \ket{0}$ to
\begin{align*}
    (\one \otimes H) \big(U_{j} \otimes \ketbra{0}{0} + U_{k} \otimes \ketbra{1}{1}\big) (\one \otimes (S H)) \ket{\iota} \otimes \ket{0} \\
    = \frac{1}{2} \big(U_{j} + i U_{k}\big) \ket{\iota} \otimes \ket{0} + \frac{1}{2} \big(U_{j} - i U_{k}\big) \ket{\iota} \otimes \ket{1}.
\end{align*}
Measuring the ancilla qubit in $\ket{0}$, in turn, implements the normalised version of $(U_{j} + i U_{k}) \ket{\iota}$ on the main register with inverse normalisation factor
\begin{align}\label{equation:NormalisationFactorImaginary}
    \big\lVert(U_{j}^{\vphantom{\dagger}} + i U_{k}^{\vphantom{\dagger}}) \ket{\iota}\big\rVert &= \sqrt{\braket{\iota | (U_{j}^{\dagger} - i U_{k}^{\dagger}) (U_{j}^{\vphantom{\dagger}} + i U_{k}^{\vphantom{\dagger}}) | \iota}} \\
    &= \sqrt{2 - 2 \Im(\braket{\iota | U_{j}^{\dagger} U_{k}^{\vphantom{\dagger}} | \iota})}
\end{align}
and probability
\begin{align}\label{equation:ProbabilityImaginary}
    q_{j k}(0) =\hspace*{-1pt} \frac{1}{4} \big\lVert(U_{j}^{\vphantom{\dagger}}\hspace*{-2pt} +\hspace*{-1pt} i U_{k}^{\vphantom{\dagger}})\hspace*{-1pt} \ket{\iota}\hspace*{-2pt}\big\rVert^{2}\hspace*{-2pt} = \frac{1}{2} \big(1\hspace*{-1pt} -\hspace*{-1pt} \Im(\braket{\iota | U_{j}^{\dagger} U_{k}^{\vphantom{\dagger}} | \iota})\big).
\end{align}
In order to obtain the remaining estimates, we execute the respective LCU operations plus subsequent measurement of the ancilla qubit.
If we measure the ancilla in $\ket{0}$, we perform a measurement on the main register's qubits in the CB, which yields a bit string as output.
The empirical averages of the observables $\one, \one - \projection_{\solspace}$, and $C$ over the bit string sample divided by the squared normalisation factors~\eqref{equation:NormalisationFactorReal} and \eqref{equation:NormalisationFactorImaginary} is the estimated expectation value of these observables w.r.t to the unnormalised states $(U_{j} + U_{k}) \ket{\iota}$ and $(U_{j} + i U_{k}) \ket{\iota}$, respectively.
Note that when forming the respective empirical averages, we have to divide by the number of zero measurements which is given by $p_{j k}(0) \cdot m$ and $q_{j k}(0) \cdot m$, respectively, where $m$ is the number of total repetitions of the LCU operation.
According to \eqref{equation:ProbabilityReal} and \eqref{equation:ProbabilityImaginary}, the probabilities are proportional to the squared inverse normalisation factors such that those cancel each other out, leaving a prefactor of $4 / m$ for the empirical averages.
The entire protocol is summarised in \autoref{protocol:EstimateSampleMatrices} as \texttt{samplemat} which inputs a collection of search unitaries $\{U_{j}\}_{j = 1}^{\ell}$, an initial state $\ket{\iota}$, the classical objective function $c$, a classical feasibility oracle $d$ which maps $\bm{b} \in \solset$ to one and $\bm{b} \in \bit^{n} \setminus \solset$ to zero, as well as the number of samples $m$ that are used for each matrix entry.
\begin{algorithm}
    \caption{\label{protocol:EstimateSampleMatrices}\texttt{samplemat}($\{U_{j}\}_{j = 1}^{\ell},\, \ket{\iota},\, c,\, d,\, m$)}
    Initialise $\mathbf{F}$, $\mathbf{G}$, and $\mathbf{H}$ as complex $\ell \times \ell$ matrices\;
    \For{$j = 1, \ldots, \ell$}{
        $\mathbb{X}_{j}$ $\coloneqq$ list of $m$ bit string samples from $U_{j} \ket{\iota}$\;
        $\mathbf{F}_{j j}\hspace*{-1pt} \coloneqq\hspace*{-1pt} 1$\;
        $\mathbf{G}_{j j}\hspace*{-1pt} \coloneqq\hspace*{-1pt} \frac{1}{m} \sum_{\bm{b} \scriptin \mathbb{X}_{j}} (1 - d(\bm{b}))$\;
        $\mathbf{H}_{j j}\hspace*{-1pt} \coloneqq\hspace*{-1pt} \frac{1}{m} \sum_{\bm{b} \scriptin \mathbb{X}_{j}} c(\bm{b})$\;
        \For{$k = 1, \ldots, j - 1$}{
            $\mathbb{X}_{j k}\hspace*{-1pt} \coloneqq\hspace*{-1pt}$ list of $m$ samples from $\Lambda_{\Re}^{(j k)} \ket{\iota} \otimes \ket{0}$\;
            $\mathbb{Y}_{j k}\hspace*{-1pt} \coloneqq\hspace*{-1pt}$ list of $m$ samples from $\Lambda_{\Im}^{(j k)} \ket{\iota} \otimes \ket{0}$\;
            $\mathbf{F}_{j k}\hspace*{-1pt} \coloneqq\hspace*{-1pt} \frac{1}{2} \Big(\hspace*{-1pt}\frac{4}{m}\hspace*{-1pt} \big(\hspace*{-1pt}\sum_{(\bm{b}, b) \scriptin \mathbb{X}_{j k}}\hspace*{-4pt} (1\hspace*{-1pt} -\hspace*{-1pt} b)\big)\hspace*{-1pt} -\hspace*{-1pt} \mathbf{F}_{j j}\hspace*{-1pt} -\hspace*{-1pt} \mathbf{F}_{k k}\Big)$
            $\quad\quad - \frac{i}{2}\Big(\hspace*{-1pt}\frac{4}{m}\hspace*{-1pt} \big(\hspace*{-1pt}\sum_{(\bm{b}, b) \scriptin \mathbb{Y}_{j k}}\hspace*{-4pt} (1\hspace*{-1pt} -\hspace*{-1pt} b)\big)\hspace*{-1pt} -\hspace*{-1pt} \mathbf{F}_{j j}\hspace*{-1pt} -\hspace*{-1pt} \mathbf{F}_{k k}\Big)$\;
            $\mathbf{F}_{k j}\hspace*{-1pt} \coloneqq\hspace*{-1pt} \overline{\mathbf{F}_{j k}}$\;
            $\mathbf{G}_{j k}\hspace*{-1pt} \coloneqq\hspace*{-1pt} \frac{2}{m} \big(\hspace*{-1pt}\sum_{(\bm{b}, b) \scriptin \mathbb{X}_{j k}}\hspace*{-4pt} (1\hspace*{-1pt} -\hspace*{-1pt} d(\bm{b})) (1\hspace*{-1pt} -\hspace*{-1pt} b)\big)$
            $\quad\quad - \frac{2 i}{m} \big(\hspace*{-1pt}\sum_{(\bm{b}, b) \scriptin \mathbb{Y}_{j k}}\hspace*{-4pt} (1\hspace*{-1pt} -\hspace*{-1pt} d(\bm{b})) (1\hspace*{-1pt} -\hspace*{-1pt} b)\big)$
            $\quad\quad + \Big(\frac{i}{2}\hspace*{-1pt} -\hspace*{-1pt} \frac{1}{2}\Big) (\mathbf{G}_{j j}\hspace*{-1pt} +\hspace*{-1pt} \mathbf{G}_{k k})$\;
            $\mathbf{G}_{k j}\hspace*{-1pt} \coloneqq\hspace*{-1pt} \overline{\mathbf{G}_{j k}}$\;
            $\mathbf{H}_{j k}\hspace*{-1pt} \coloneqq\hspace*{-1pt} \frac{1}{2} \Big(\hspace*{-1pt}\frac{4}{m}\hspace*{-1pt} \big(\hspace*{-1pt}\sum_{(\bm{b}, b) \scriptin \mathbb{X}_{j k}}\hspace*{-4pt} c(\bm{b}) (1\hspace*{-1pt} -\hspace*{-1pt} b)\big)\hspace*{-1pt} -\hspace*{-1pt} \mathbf{H}_{j j}\hspace*{-1pt} -\hspace*{-1pt} \mathbf{H}_{k k}\hspace*{-1pt}\Big)$
            $\quad\ \ - \frac{i}{2}\Big(\hspace*{-1pt}\frac{4}{m}\hspace*{-1pt} \big(\hspace*{-1pt}\sum_{(\bm{b}, b) \scriptin \mathbb{Y}_{j k}}\hspace*{-4pt} c(\bm{b}) (1\hspace*{-1pt} -\hspace*{-1pt} b)\big)\hspace*{-1pt} -\hspace*{-1pt} \mathbf{H}_{j j}\hspace*{-1pt} -\hspace*{-1pt} \mathbf{H}_{k k}\hspace*{-1pt}\Big)$\;
            $\mathbf{H}_{k j}\hspace*{-1pt} \coloneqq\hspace*{-1pt} \overline{\mathbf{H}_{j k}}$\;
        }
    }
    \Return{($\mathbf{F}, \mathbf{G}, \mathbf{H}$)}\;
\end{algorithm}

\subsection{\label{subsection:GeneralisedEigenvalueProblem}Generalised eigenvalue problem}

We proceed by translating \eqref{equation:GenQCPMatrixParameterOptimisation} into a solver-friendlier form: a \emph{generalised eigenvalue problem} (GEP).
This transformation has already been discussed for a single constraint.
We extend the transformation to the doubly-constrained case and fill some gaps present in the literature on GEP derivations with mathematical rigour, that persisted in the proposals by McClean~\textit{et~al.}~\cite{Mcclean2017HybridQuantumClassicalHierarchyForMitigationOfDecoherenceAndDeterminationOfExcitedStates}, Huggings~\textit{et~al.}~\cite{Huggins2020ANonOrthogonalVariationalQuantumEigensolver}, and Binkowski~\textit{et~al.}~\cite{Binkowski2025FromBarrenPlateausThroughFertileValleysConicExtensionsOfParameterisedQuantumCircuits}.
The following assumption will be central for further derivations.
\begin{assumption}\label{assumption:LinearIndependence}
    The state $\ket{\iota} \in \hil$ and unitaries $\{U_{j}\}_{j = 1}^{\ell}$ are chosen in such a way that $\{U_{j} \ket{\iota} \defcolon j = 1, \ldots, \ell\}$ constitutes a linearly independent set of states.
\end{assumption}
Note that, as discussed after \autoref{lemma:ConstructingSampleMatricesIsCompletelyPositive}, this \hyperref[assumption:LinearIndependence]{Assumption} implies that $\mathbf{F}$ and $\mathbf{H}$ will be positive-definite, i.e., positive-semidefinite and invertible.
This, in turn, constitutes a crucial assumption for the derivation of the GEP.

Quadratic optimisation problems such as \eqref{equation:GenQCPMatrixParameterOptimisation} are well-studied objects in the optimisation theory literature with many results on solvability conditions and duality structure~\cite{Thoai2000DualityBoundMethodForTheGeneralQuadraticProgrammingProblemWithQuadraticConstraints}.
In particular,~\cite{Nguyen2022StrongDualityForGeneralQuadraticProgramsWithQuadraticEqualityConstraints} provides a rich collection of tools for investigating strong duality of quadratic problems of which we merely require a small subset.

The first step is to formulate the complex problem \eqref{equation:GenQCPMatrixParameterOptimisation} over a real vector space by canonically identifying $\C^{\ell}$ with $\R^{2 \ell}$ via $\C \ni \bm{\alpha} = \Re(\bm{\alpha}) + i \Im(\bm{\alpha})\, \widehat{=}\, (\Re(\bm{\alpha}), \Im(\bm{\alpha})) \in \R^{2 \ell}$.
Additionally, we introduce the real $2 \ell \times 2 \ell$ symmetric block matrices
\begin{align}\label{equation:RealSymmetricMatrices}
    \bm{F} \coloneqq \begin{pmatrix}
        \Re(\mathbf{F})& -\Im(\mathbf{F}) \\
        \Im(\mathbf{F})& \hphantom{-}\Re(\mathbf{F})
    \end{pmatrix},
\end{align}
as well as $\bm{G}$ and $\bm{H}$ analogously constructed from $\mathbf{G}$ and $\mathbf{H}$, respectively.
Note that $\bm{F}$ and $\bm{H}$ are again positive-definite while $\bm{G}$ is guaranteed to be positive-semidefinite.
We can then readily rewrite \eqref{equation:GenQCPMatrixParameterOptimisation} as
\begin{align}\label{equation:GenQCPParameterOptimisationReal}
    \begin{split}
        &\min_{\bm{x} \deepscriptin \R^{2 \ell}} \bm{x}^{\transpose} \bm{H} \bm{x} \\
        &\mathhphantomdisplay{\min_{\bm{x} \deepscriptin \R^{2 \ell}}}{\ \, \text{s.t.}}\ \mathhphantomdisplay{\bm{x}^{\transpose} \bm{G} \bm{x}}{\bm{x}^{\transpose} \bm{F} \bm{x}} = 1 \\
        &\hphantom{\min_{\bm{x} \deepscriptin \R^{2 \ell}}}\ \bm{x}^{\transpose} \bm{G} \bm{x} = 0.
    \end{split}
\end{align}

Since $\bm{G}$ is positive-semidefinite, it holds, given an $\bm{x} \in \R^{2 \ell}$, that $\bm{x}^{\transpose} \bm{G} \bm{x} = 0$ if and only if $\bm{G} \bm{x} = \bm{0}$.
Thus, we effectively optimise over $\ker(\bm{G})$ which we can absorb implicitly into the problem's domain.
To formulate the problem again in a concrete coordinate space $\R^{m}$ where $m = \dim(\ker(\bm{G}))$, we may conjugate both $\bm{F}$ and $\bm{H}$ with the basis transform $\bm{B}_{\bm{G}}$ into the orthogonal eigenbasis of $\bm{G}$ and then execute the sandwiching with the projection onto $\ker(\bm{G})$ by deleting all rows/columns that correspond to non-kernel eigenvectors of $\bm{G}$.
$\bm{F}$ and $\bm{H}$ both being positive-definite readily implies that also $\bm{B}_{\bm{G}}^{\vphantom{-1}} \bm{F} \bm{B}_{\bm{G}}^{-1}$ and $\bm{B}_{\bm{G}}^{\vphantom{-1}} \bm{H} \bm{B}_{\bm{G}}^{-1}$ are positive definite.
By the same argument as in the proof of \autoref{lemma:ConstructingSampleMatricesIsCompletelyPositive}, their restrictions to the $\R^{m}$-subspace $\tilde{\bm{F}} \coloneqq (\projection_{\R^{m}} \bm{B}_{\bm{G}}^{\vphantom{-1}}\bm{F} \bm{B}_{\bm{G}}^{-1} \projection_{\R^{m}})\vert_{\R^{m}}$ and analogously defined $\tilde{\bm{H}}$ are positive-definite.
It then holds that \eqref{equation:GenQCPParameterOptimisationReal} is equivalent to
\begin{align}\label{equation:GenQCPParameterOptimisationRealImplicit}
    \begin{split}
        &\min_{\bm{x} \scriptin \R^{m}} \bm{x}^{\transpose} \tilde{\bm{H}} \bm{x} \\
        &\mathhphantomdisplay{\min_{\bm{x} \deepscriptin \R^{m}}}{\ \, \text{s.t.}}\ \bm{x}^{\transpose} \tilde{\bm{F}} \bm{x} = 1.
    \end{split}
\end{align}
The Lagrange function for \eqref{equation:GenQCPParameterOptimisationRealImplicit} is
\begin{align}\label{equation:LagrangeFunction}
    \begin{split}
        g : \R^{m} \times \R &\rightarrow \R \\
        (\bm{x}, \lambda) &\mapsto \bm{x}^{\transpose} \tilde{\bm{H}} \bm{x} + \lambda (1 - \bm{x}^{\transpose} \tilde{\bm{F}} \bm{x}).
    \end{split}
\end{align}
and, after rearranging the terms of \eqref{equation:LagrangeFunction}, we find that the Lagrangian dual problem has the form
\begin{align}\label{equation:LagrangianDualProblem}
    \max_{\lambda \scriptin \R} \min_{\bm{x} \deepscriptin \R^{m}} g(\bm{x}, \lambda) = \max_{\lambda \scriptin \R} \bigg(\hspace*{-2pt}\lambda + \min_{\bm{x} \deepscriptin \R^{m}}\hspace*{-1pt} \bm{x}^{\transpose}\hspace*{-1pt} \Big(\hspace*{-1pt}\tilde{\bm{H}}\hspace*{-1pt} -\hspace*{-1pt} \lambda \tilde{\bm{F}}\hspace*{-1pt}\Big) \bm{x}\hspace*{-1pt}\bigg).
\end{align}
Note that for all $\lambda \in \R$ for which the symmetric matrix $\tilde{\bm{H}} - \lambda \tilde{\bm{F}}$ is not positive-semidefinite, we can find a direction $\hat{\bm{x}} \in \R^{m}$ with $\hat{\bm{x}}^{\transpose} (\tilde{\bm{H}} - \lambda \tilde{\bm{F}}) \hat{\bm{x}} < 0$ such that minimising over all $\bm{x} \in \R^{m}$ is unbounded below.
In contrast, in case of positive-semidefiniteness, the minimal possible value is zero which is attained by choosing, e.g., $\bm{x} = \bm{0}$.
This shows that the Lagrangian dual problem \eqref{equation:LagrangianDualProblem} has the equivalent form
\begin{align}\label{equation:ConstrainedLagrangianDualProblem}
    \begin{split}
        &\max_{\lambda \scriptin \R} \lambda \\
        &\mathhphantomdisplay{\max_{\lambda \scriptin \R}}{\ \text{s.t. }}\ \tilde{\bm{H}} - \lambda \tilde{\bm{F}} \succeq 0.
    \end{split}
\end{align}
An adaptation of~\cite[Lemma 2.1]{Nguyen2022StrongDualityForGeneralQuadraticProgramsWithQuadraticEqualityConstraints} translates the positive-semidefiniteness constraint into a system of inequalities.
\begin{lemma}\label{lemma:InequalitySystem}
    Let $\bm{A} \in \R^{d \times d}$ be symmetric and $\bm{B} \in \R^{d \times d}$ be positive-semidefinite.
    Then $\bm{A} \succeq 0$ is equivalent to the inequality system
    \begin{align}
        &\bm{x}^{\transpose} \bm{A} \bm{x} \geq 0 \text{ for all } \bm{x} \in \R^{d} \text{ satisfying } \bm{x}^{\transpose} \bm{B} \bm{x} = 1, \label{equation:InequalitySystem1} \\
        &\bm{x}^{\transpose} \bm{A} \bm{x} \geq 0 \text{ for all } \bm{x} \in \R^{d} \text{ satisfying } \bm{x}^{\transpose} \bm{B} \bm{x} = 0. \label{equation:InequalitySystem2}
    \end{align}
\end{lemma}

\begin{proof}
    $\bm{A} \succeq 0$ trivially implies both \eqref{equation:InequalitySystem1} and \eqref{equation:InequalitySystem2}.
    For the opposite direction, let $\bm{A} \not\succeq 0$, i.e., there exists $\bm{y} \in \R^{d}$ such that $\bm{y}^{\transpose} \bm{A} \bm{y} < 0$.
    If either $\bm{y}^{\transpose} \bm{B} \bm{y} = 1$ or $\bm{y}^{\transpose} \bm{B} \bm{y} = 0$, we are already done.
    If neither of these conditions are fulfilled, consider $\bm{x} \coloneqq \bm{y} / \sqrt{\bm{y}^{\transpose} \bm{B} \bm{y}}$ which, by construction, fulfils $\bm{x}^{\transpose} \bm{B} \bm{x} = 1$, but also $\bm{x}^{\transpose} \bm{A} \bm{x} < 0$, i.e., \eqref{equation:InequalitySystem1} is violated.
\end{proof}

Applying \autoref{lemma:InequalitySystem} to the family $\tilde{\bm{H}} - \lambda \tilde{\bm{F}}$ of symmetric matrices, $\tilde{\bm{F}} \succ 0$ translates the positive-semidefiniteness constraint into two conditions:
For all $\bm{x} \in \R^{m}$, satisfying $\bm{x}^{\transpose} \tilde{\bm{F}} \bm{x} = 1$, it must hold that
\begin{align*}
    \bm{x}^{\transpose} (\tilde{\bm{H}} - \lambda \tilde{\bm{F}}) \bm{x} \geq 0 \iff \bm{x}^{\transpose} \tilde{\bm{H}} \bm{x} \geq \lambda.
\end{align*}
For all $\bm{x} \in \R^{m}$, satisfying $\bm{x}^{\transpose} \tilde{\bm{F}} \bm{x} = 0$, it must hold that
\begin{align}\label{equation:InequalitySystemApplied2}
    \bm{x}^{\transpose} (\tilde{\bm{H}} - \lambda \tilde{\bm{F}}) \bm{x} \geq 0 \Leftrightarrow \bm{x}^{\transpose} \tilde{\bm{H}} \bm{x} \geq 0.
\end{align}
Clearly, \eqref{equation:InequalitySystemApplied2} is trivially satisfied since $\tilde{\bm{H}} \succ 0$.
This leaves us with the equivalent problem
\begin{align*}
    \max\{\lambda \in \R\, \colon \lambda \leq \bm{x}^{\transpose} \tilde{\bm{H}} \bm{x}\, \forall\, \bm{x} \in \R^{m} \text{ satisfying } \bm{x}^{\transpose} \tilde{\bm{F}} \bm{x} = 1\},
\end{align*}
which clearly has the same optimal value as \eqref{equation:GenQCPParameterOptimisationRealImplicit} and hence as \eqref{equation:GenQCPMatrixParameterOptimisation}.
This shows strong duality for the primal-dual pair \eqref{equation:GenQCPParameterOptimisationRealImplicit}--\eqref{equation:LagrangianDualProblem}, provided that the primal problem is solvable.
The latter, however, follows from the fact that the original problem \eqref{equation:GenQCPParameterOptimisation} is always solvable with $\bm{\alpha} = (0, \ldots, 0, 1)$.

After having established strong duality, it remains to show that the Lagrangian dual problem admits an equivalent formulation as a GEP.
For this final derivation, we start with the form \eqref{equation:ConstrainedLagrangianDualProblem} of the dual problem.
First note that the linear function $l \coloneqq [\lambda \mapsto \tilde{\bm{H}} - \lambda \tilde{\bm{F}}]$ is analytic and monotonically decreasing since $(\tilde{\bm{H}} - \lambda \tilde{\bm{F}}) - (\tilde{\bm{H}} - \mu \tilde{\bm{F}}) = (\mu - \lambda) \tilde{\bm{F}} \succeq 0$ for $\mu \ge \lambda$.
Thus, if $l(\lambda_{0}) = \tilde{\bm{H}} - \lambda_{0} \tilde{\bm{F}} \succeq 0$ for a given $\lambda_{0} \in \R$, then also $l(\lambda) \succeq 0$ for all $\lambda \leq \lambda_{0}$.
Consider further the set $D \coloneqq \{\lambda \in \R \defcolon l(\lambda) \succeq 0 \text{ and } l(\lambda) \not\succ 0\}$, i.e.\ the set of all real numbers $\lambda$ for which $l(\lambda)$ is positive-semidefinite, but not positive-definite.
Since \eqref{equation:ConstrainedLagrangianDualProblem} is feasible we know that $\{l(\lambda) \defcolon l(\lambda) \succeq 0\}$ is non-empty, and from the fact that \eqref{equation:ConstrainedLagrangianDualProblem} is bounded above, we deduce that for high enough $\lambda$, $l(\lambda)$ will fail to be positive-semidefinite, that is the complement $\{l(\lambda) \defcolon l(\lambda) \not\succeq 0\}$ is also non-empty.
The continuity of $l$ and the inherited continuity of $l$'s eigenvalue curves~\cite{Kato1995PerturbationTheoryForLinearOperators} then implies that $\{l(\lambda) \defcolon l(\lambda) \succeq 0 \text{ and } l(\lambda) \not\succ 0\}$ is non-empty as well, hence so is its preimage $D$.
Therefore, we can choose $\lambda_{0} \in D$.
Since both $\tilde{\bm{F}}$ and $\tilde{\bm{H}}$ are positive-definite, $l(\lambda)$ will be positive-definite for all $\lambda \leq 0$, i.e.\ it holds that $D \cap (-\infty, 0] = \emptyset$ and therefore that $\lambda_{0} > 0$.
Moreover, by~\cite[Problem 6.2]{Kato1995PerturbationTheoryForLinearOperators}, the function
\begin{align*}
    \gamma_{\min} \coloneqq [\lambda \mapsto \min\{\gamma \in \R \defcolon \gamma \text{ is an eigenvalue of } l(\lambda)\}]
\end{align*}
is concave.
For our chosen $\lambda_{0}$ it holds that $\gamma_{\min}(\lambda_{0}) = 0$;
positive definiteness of $\tilde{\bm{H}}$ especially implies that $\gamma_{\min}(0) > 0$.
For all $\lambda_{1}, \lambda_{2} \in [0, \infty)$ with $\lambda_{2} > \lambda_{1}$ it holds that there exists an $\alpha \in (0, 1)$ such that $\lambda_{1} = \alpha \lambda_{2}$, and therefore, since $\gamma_{\min}$ is concave, that
\begin{align}\label{equation:StrictConcaveInequality}
    \begin{split}
        \gamma_{\min}(\lambda_{1}) &= \gamma_{\min}(\alpha \lambda_{2} + (1 - \alpha) 0) \\
        & \geq \alpha \gamma_{\min}(\lambda_{2})\hspace*{-1pt} +\hspace*{-1pt} (1\hspace*{-1pt} -\hspace*{-1pt} \alpha) \gamma_{\min}(0) > \alpha \gamma_{\min}(\lambda_{2}).
    \end{split}
\end{align}
Applied to $\lambda_{1} = \lambda_{0}$ and an arbitrary $\lambda_{2} > \lambda_{0}$, this shows $\gamma_{\min}(\lambda_{2}) < 0$ since, otherwise by~\eqref{equation:StrictConcaveInequality}, $0 = \gamma_{\min}(\lambda_{0}) > 0$, raising a contradiction.
This result, in turn, shows that $D$ only consists of the single point $\lambda_{0}$ since all values larger than $\lambda_{0}$ give rise to an operator $l(\lambda)$ with at least one negative eigenvalue.
In summary, we have shown that the feasible set of \eqref{equation:ConstrainedLagrangianDualProblem} consists only of a single point $\lambda_{0}$ which fulfils that $l(\lambda_{0})$ is positive-semidefinite with at least one eigenvalue being zero.
Hence there exists a vector $\bm{x} \in \R^{m} \setminus \{\bm{0}\}$ such that
\begin{align}\label{equation:GeneralisedEigenvalueProblem}
    (\tilde{\bm{H}} - \lambda_{0} \tilde{\bm{F}}) \bm{x} = \bm{0}\, \Leftrightarrow\, \tilde{\bm{H}} \bm{x} = \lambda_{0} \tilde{\bm{F}} \bm{x}.
\end{align}
which is precisely the formulation of a GEP.
A GEP-solver applied to the tuple $(\tilde{\bm{H}}, \tilde{\bm{F}})$ will generally output several generalised eigenvalues and eigenvectors.
By the above arguments, $\lambda_{0}$ will be the smallest eigenvalue among them, as for every larger generalised eigenvalue $\tilde{\lambda} > \lambda_{0}$, the operator $l(\tilde{\lambda})$ is not positive-semidefinite.
We normalise the generalised eigenvector $\bm{x}_{0}$ associated to $\lambda_{0}$ by $(\bm{x}_{0}^{\transpose} \tilde{\bm{F}} \bm{x}_{0}^{\vphantom{\transpose}})^{-1 / 2}$.
By the definition of a generalised eigenvector, we find that $\bm{x}_{0}$ is indeed the optimiser of \eqref{equation:ConstrainedLagrangianDualProblem}:
\begin{align*}
    \bm{x}_{0}^{\transpose} \tilde{\bm{H}} \bm{x}_{0}^{\vphantom{\transpose}} = \bm{x}_{0}^{\transpose} (\lambda_{0} \tilde{\bm{F}} \bm{x}_{0}^{\vphantom{\transpose}}) = \lambda_{0} \bm{x}_{0}^{\transpose} \tilde{\bm{F}} \bm{x}_{0}^{\vphantom{\transpose}} = \lambda_{0}.
\end{align*}

Following just established strong duality, undoing the basis change into $\bm{G}$'s eigenbasis, inverting the complex-to-real transformation \eqref{equation:RealSymmetricMatrices}, and utilising the definitions of the sample matrices, we have thus concluded the following theorem.
\begin{theorem}\label{theorem:GeneralisedEigenvalueProblem}
    The minimiser $\bm{\alpha}_{0}$ and the minimum $\epsilon_{0}$ of \eqref{equation:GenQCPParameterOptimisation} are given by
    \begin{align}\label{equation:Minimiser}
        \begin{split}
            \bm{\alpha_{0}} &= \big[\bm{B}_{\bm{G}}^{-1}(\bm{x}_{0}, 0, \ldots, 0)\big]_{1, \ldots, \ell} \\
            &\quad + i \big[\bm{B}_{\bm{G}}^{-1}(\bm{x}_{0}, 0, \ldots, 0)\big]_{\ell + 1, \ldots, 2 \ell}
        \end{split}
    \end{align}
    and the minimal generalised eigenvalue $\lambda_{0}$ of the generalised eigenvalue problem \eqref{equation:GeneralisedEigenvalueProblem}, where $\bm{x}_{0}$ is the generalised eigenvector associated to $\lambda_{0}$.
\end{theorem}

Most importantly, \autoref{theorem:GeneralisedEigenvalueProblem} rigorously reduces the complex quadratic optimisation problem \eqref{equation:GenQCPParameterOptimisation} to an efficiently constructible and solvable, real-valued GEP.

\subsection{\label{subsection:LCUImplementation}LCU implementation}

Solving \eqref{equation:GeneralisedEigenvalueProblem} and applying the transformation \eqref{equation:Minimiser}, we find the parameter vector $\bm{\alpha}$ such that $\mathcal{M}_{\bm{\alpha}} \ket{\iota}$ constitutes a state minimising the expectation value of $C$ among the parameterised ansatz class.
What remains to show is how to prepare the given states on a QPU.
The procedure was already explained in~\cite{Binkowski2025FromBarrenPlateausThroughFertileValleysConicExtensionsOfParameterisedQuantumCircuits} and applies analogously to our generalisation.
In this subsection, we merely provide some additional details of the LCU construction.

For $\ell$ search unitaries, we introduce a $\lceil\log_{2}(\ell)\rceil$-qubit ancilla register $\hil_{\ancilla}$ on which we prepare the quantum state
\begin{align*}
    \ket{\psi}_{\ancilla} \coloneqq \frac{(\sqrt{\bm{\alpha}}, \bm{0})}{\sqrt{\lVert\bm{\alpha}\rVert_{1}}}
\end{align*}
where the padding with $2^{\lceil \log_{2}(\ell)\rceil} - \ell$ many zeros is necessary in order to obtain a valid $\lceil \log_{2}(\ell)\rceil$-qubit state vector.
The state preparation of $\ket{\psi}_{\ancilla}$ may be executed via, e.g., the Grover-Rudolph algorithm~\cite{Ramacciotti2024SimpleQuantumAlgorithmToEfficientlyPrepareSparseStates} with $\mathcal{O}(\ell \log_{2} \ell)$-many gates.
Following the state preparation of $\ket{\psi}_{\ancilla}$, we apply the unitary
\begin{align*}
    \mathcal{U} \coloneqq \sum_{j = 1}^{\ell} U_{j} \otimes \ketbra{j}{j}_{\ancilla} + \one \otimes \Bigg(\sum_{j = \ell + 1}^{2^{\lceil \log_{2}(\ell)\rceil}} \ketbra{j}{j}_{\ancilla}\Bigg)
\end{align*}
to the product state $\ket{\iota} \otimes \ket{\psi}_{\ancilla}$, where we label the ancilla register's CB states with integers $j$ rather than bit strings.
Note that since $\ket{\psi}_{\ancilla}$ does not have any components in the CB states with index larger than $\ell$, the second term does not contribute when applying $\mathcal{U}$ to $\ket{\iota} \otimes \ket{\psi}_{\ancilla}$.
Subsequently, we conduct a binary measurement in direction of
\begin{align*}
    \ket{\xi}_{\ancilla} \coloneqq \frac{\big(\overline{\sqrt{\bm{\alpha}}}, \bm{0}\big)}{\sqrt{\lVert\bm{\alpha}\rVert_{1}}}
\end{align*}
by first applying its inverse state preparation and then conducting a measurement in the CB on the ancilla register, distinguishing between $\ket{\bm{0}}_{\ancilla}$ and all other states.
All these steps taken together have the following action
\begin{align*}
    &\big(\one \otimes \bra{\xi}_{\ancilla}\big) \mathcal{U} \big(\ket{\iota} \otimes \ket{\psi}_{\ancilla}\big) \\
    &= \lVert\bm{\alpha}\rVert_{1}^{-1}\hspace*{-5pt} \sum_{i, j, k = 1}^{\ell}\hspace*{-5pt} \big(\one \otimes \sqrt{\alpha_{i}} \bra{i}_{\ancilla}\hspace*{-1pt}\big) \big(U_{j} \otimes \ketbra{j}{j}_{\ancilla}\hspace*{-1pt}\big) \big(\hspace*{-1pt}\ket{\iota} \otimes \sqrt{\alpha_{k}} \ket{k}_{\ancilla}\hspace*{-1pt}\big) \\
    &= \lVert\bm{\alpha}\rVert_{1}^{-1}\hspace*{-2pt} \sum_{i, j = 1}^{\ell}\hspace*{-2pt} \big(\one \otimes \sqrt{\alpha_{i}} \bra{i}_{\ancilla}\big) \big(U_{j} \ket{\iota} \otimes \sqrt{\alpha_{j}} \ket{j}_{\ancilla}\big) \\
    &= \lVert\bm{\alpha}\rVert_{1}^{-1} \sum_{j = 1}^{\ell} \alpha_{j} U_{j} \ket{\iota} = \lVert\bm{\alpha}\rVert_{1}^{-1} M_{\bm{\alpha}} \ket{\iota}.
\end{align*}
The normalisation constraint $\lVert M_{\bm{\alpha}} \ket{\iota}\rVert = 1$, implies that the probability of measuring the all-zero bit string upon measurement on the ancilla register, i.e., the success probability of implementing the LCU correctly, is precisely $\lVert\bm{\alpha}\rVert_{1}^{-2}$.
A single successful application of the LCU update therefore implements the $\mathcal{M}_{\bm{\alpha}} \ket{\iota}$ on the main register.

\subsection{\label{subsection:FeasibilityPreservingUnitaries}Feasibility-preserving unitaries}

While a key advantage of our construction is the ability to ensure that a successful LCU step is feasibility-preserving even if the comprising search unitaries $U_j$ are not, we briefly discuss simplifications arising in the case when a set of such feasibility-preserving unitaries \emph{is} available for the considered CCOP.
Recall that $U_j$ is feasibility-preserving if $U_j(\solspace) \subset \solspace$, which, by $\ket{\iota} \in \solspace$, implies $\{U_{j} \ket{\iota} \defcolon j = 1, \ldots, \ell\} \subset \solspace$.
Since $\ker(\one - \projection_{\solspace}) = \solspace$, the sample matrix $\mathbf{G}$ (and therefore $\bm{G}$) is the null matrix and the constraint $\bm{\alpha}^\dagger \mathbf{G} \bm{\alpha} = 0$ (and $\bm{x}^\dagger \bm{G} \bm{x} = 0$) is invariably satisfied for all $\bm{\alpha}$ (and $\bm{x}$).
The problem~\eqref{equation:GenQCPParameterOptimisation} is reduced to a higher-dimensional variant of the original problem~\eqref{equation:QCPParameterOptimisation}.
The derivation of strong duality in \autoref{subsection:GeneralisedEigenvalueProblem} still holds for the then trivial case of $m = 2 \ell$ and $\bm{B}_{\bm{G}} = \one_m$, and \autoref{theorem:GeneralisedEigenvalueProblem} implies that the minimiser and minimum of \eqref{equation:GenQCPParameterOptimisation} are found by determining the minimal generalised eigenvalue and associated generalised eigenvector of the $2 \ell$-dimensional problem
\begin{align}
    \bm{H} \bm{x} = \lambda \bm{F} \bm{x}.
\end{align}

We further note that for many CCOPs the feasible subspace $\solspace$ is much smaller dimension than the entire Hilbert space $\hil$.
Therefore, choosing search unitaries at random will result in mostly fully infeasible directions $U_{j} \ket{\iota}$.
However, if the LCU sequence as a whole is capable of extracting a feasible state from intermediate, infeasible states, the method may traverse $\solspace^{\perp}$.
Ultimately, it is desired to transfer knowledge from constructing feasibility-preserving QAOA-mixers to our framework, but also unitaries which trade in strict feasibility-preservation for reduced circuit depth or cheaper gates are well-suited.
One prototypical example is to replace costly controlled-swap gates in permutation-based mixer constructions~\cite{Binkowski2025SymmetryBasedQuantumAlgorithmsForOpenShopSchedulingWithHardConstraints} with easier to implement CNOT gates which still offer valuable state transitions.
Another is to incorporate mixers of exhaustively-parametrised quantum circuits~\cite{Schwiering2026_ExhaustiveAndFeasibleParametrisationWithApplicationsToTheTravellingSalespersonProblem} to lift the implementation probability of individual channels.

\subsection{\label{subsection:RobustnessOfPurityAgainstNoise}Robustness of purity against noise}

The above scheme is formulated for any device that solely operates on pure quantum states, representable by a state vector $\ket{\phi} \in \hil$.
In reality however, we typically operate on \emph{mixed} quantum states, described by the set of density operators $\mathfrak{S} \coloneqq \{ \rho \in \lo(\hil): \rho \succeq 0\ \text{and}\ \tr[\rho] = 1 \}$, where the pure states are exactly the rank-one projections $\rho = \ketbra{\phi}{\phi} \in \mathfrak{S}(\hil)$.
This situation arises when we deploy a quantum computer with access to only imperfect gate operations, e.g., a NISQ device, to solve the optimisation task~\eqref{equation:GenQCPParameterOptimisation}.
This section will prove a certain level of noise resistance to our framework alleviating the effects of noise to the classical parameter optimisation.
Namely, we prove that the suitably reformulated \eqref{equation:GenQCPParameterOptimisation} remains feasible and that its solution remains vector-valued (rather than matrix-valued) even if the pure initial state $\ket{\iota}$ is replaced with a \emph{feasible} mixed state $\rho_{\text{in}} \in \mathfrak{S}(\hil)$, $\supp(\rho) \subseteq \solspace$.
That is, even under the influence of noise in the feasible subspace, it suffices to prepare a pure state in the ancilla register in order to update the main register's mixed state to the optimal solution found in the parameterised ansatz class.

Note that by switching to the density operator formalism, the task of minimising the objective Hamiltonian's expectation value becomes a constrained semidefinite program over the cone of linear, positive-semidefinite operators on the main register's Hilbert space $\PSD(\hil) \supset \mathfrak{S}(\hil)$.
The normalisation condition is now to be expressed in terms of having a unit trace.
Similarly, orthogonality to the infeasible subspace $\solspace^{\perp}$ is expressed by a vanishing trace condition:
\begin{align}\label{equation:CCOPSDP}
    \begin{split}
        &\min_{\rho \in \PSD(\hil)} \tr[\rho C] \\
        &\mathhphantomdisplay{\min_{\rho \in \PSD(\hil)}}{\quad \ \text{s.t.}}\, \tr[\rho] = 1 \\
        &\hphantom{\min_{\rho \in \PSD(\hil)}}\, \tr[\rho (\one - \projection_{\solspace})] = 0.
    \end{split}
\end{align}

The initial state is $\rho_{\text{in}} \in \mathfrak{S}(\hil)$ and can be expressed as a convex combination of pure states; i.e., $\rho_{\text{in}} = \sum_{i} p_{i} \ketbra{i}{i}$.
In the context of a statistical ensemble, $p_{i}$ is the probability of the system to be prepared in the pure state $\ket{i} \in \hil$.
Recording the sample matrices for a given pure state $\ket{i} \in \hil$ is a positive map due to \autoref{lemma:ConstructingSampleMatricesIsCompletelyPositive}.
Since conic combinations of positive maps are again positive and any convex combination is notably also a conic combination, this further ensures that also
\begin{align}
    \begin{split}
        T : \lo(\hil) &\rightarrow \C^{\ell \times \ell}, \\
        A &\mapsto \Big(\tr[\rho_{\text{in}} U_{j}^{\dagger} A U_{k}^{\vphantom{\dagger}}]\Big)_{j, k = 1}^{\ell}
    \end{split}
\end{align}
is a positive map.
Therefore,
\begin{alignat*}{3}
    &\mathbf{F}_{j k} &&\coloneqq T(\one)_{j k} &&= \tr[\rho_{\text{in}} U_{j}^{\dag} U_{k}^{\vphantom{\dagger}}] \\
    &\mathbf{G}_{j k} &&\coloneqq T(\one - \projection_{\solspace})_{j k} &&= \tr[\rho_{\text{in}} U_{j}^{\dag} (\one - \projection_{\solspace}) U_{k}^{\vphantom{\dagger}}], \\
    &\mathbf{H}_{j k} &&\coloneqq T(C)_{j k} &&= \tr[\rho_{\text{in}} U_{j}^{\dag} C U_{k}^{\vphantom{\dagger}}]
\end{alignat*}
are entries of $\ell \times \ell$ positive-semidefinite matrices $\mathbf{F}, \mathbf{G}, \mathbf{H}$.
We have thus transformed the original problem to a potentially smaller SDP formulated over $\PSD_{\ell}$, whose solution gives an upper bound on the original problem:
\begin{equation}\label{equation:CCOPMomMatSDP}
    \begin{alignedat}{3}
        &\min_{\mathbf{X} \in \PSD_\ell} &&\tr[\mathbf{H} \mathbf{X}]
        \\
        & \quad \text{s.t.} &&\,\mathhphantomdisplay{\tr[\mathbf{G} \mathbf{X}]}{\tr[\mathbf{F} \mathbf{X}]} = 1,
        \\
        & &&\tr[\mathbf{G} \mathbf{X}] = 0.
    \end{alignedat}
\end{equation}

As a special case of the complex Barvinok-Pataki rank bound~\cite{Huang2007ComplexMatrixDecompositionAndQuadraticProgramming}, the optimal solution to \eqref{equation:CCOPMomMatSDP} will be attained especially at a rank-one matrix $\mathbf{X}_{\text{opt}}$.
The latter may then be translated into a parameter vector $\bm{\alpha} \in \C^{\ell}$ to reproduce the solution on the quantum device.
More generally, the optimum for \eqref{equation:CCOPSDP} is attained at a pure state.

For the unlikely case that \hyperref[assumption:LinearIndependence]{Assumption} cannot be guaranteed, we conclude with a scenario which does not require $\mathbf{F}$ to be non-singular for bringing \eqref{equation:CCOPMomMatSDP} into the form \eqref{equation:GenQCPParameterOptimisationRealImplicit} by first principles.
Consider a rank-one solution $\mathbf{X} = \bm{x} \bm{x}^{\dag}$, and let $\bm{x} = \bm{x}^{\parallel} + \bm{x}^{\perp}$, where $\bm{x}^{\parallel}$ and $\bm{x}^{\perp}$ are the projections of $\bm{x}$ onto $\ker \mathbf{F}$ and its orthogonal complement, respectively.

\begin{proposition}\label{proposition:SDPtoQCQP}
    If $0 \preceq \mathbf{H} \preceq \mathbf{F}$, \eqref{equation:CCOPMomMatSDP} is equivalent to the quadratic optimisation task \eqref{equation:GenQCPParameterOptimisation} with $\bm{x} \in (\ker \mathbf{F})^{\perp}$.
\end{proposition}

\begin{figure*}[t]
    \centering
    \begin{tikzpicture}
        \begin{groupplot}[
            group style={
                group size=2 by 1,
                horizontal sep=65pt,
            },
            width=0.48\textwidth,
            height=0.28\textwidth,
            xmin=0, xmax=61,
            xlabel={step},
            label style={font=\small},
            tick label style={font=\footnotesize},
            every axis plot/.append style={line width=0.6pt, opacity=0.8},
            ]
            \nextgroupplot[ylabel={$\braket{\phi | C | \phi}/c_{\text{opt}}$}]
            \foreach \i in {0, ..., 9} {
                \addplot[QCPblue] table[x=step, y=ratio] {data/qcp_qce26_ex\i.dat};
            }
            \nextgroupplot[ylabel={$\prod_{a=1}^{L}\norm*{\bm{\alpha}^{(a)}}_{1}^{-2}$}, ymode=log]
            \foreach \i in {0, ..., 9} {
                \addplot[QCPblue] table[x=step, y=prob] {data/qcp_qce26_ex\i.dat};
            }
        \end{groupplot}
    \end{tikzpicture}
    \caption{Approximation ratio $\braket{\phi | C | \phi}/c_{\text{opt}}$ (left) and aggregated LCU implementation probability $\prod_{a=1}^{L} \norm{\bm{\alpha}^{(a)}}_{1}^{-2}$ (right, log scale) over the course of the optimisation, for ten strongly correlated Pisinger knapsack instances with $n=16$ items each (one line per instance).}
    \label{figure:NumericalExperiments}
\end{figure*}

\begin{proof}
    For any $\bm{x} \in \ker \mathbf{F}$, $0 \preceq \mathbf{H} \preceq \mathbf{F}$ implies
    \begin{equation}
        0 \leq \bm{x}^{\dag} (\mathbf{F} - \mathbf{H}) \bm{x} = - \bm{x}^{\dag} \mathbf{H} \bm{x} \leq 0
    \end{equation}
    and thus, $\ker \mathbf{F} \subseteq \ker \mathbf{H}$.
    It therefore holds that
    \begin{align*}
        \tr[\mathbf{H} \mathbf{X}] &= (\bm{x}^{\perp})^{\dag} \mathbf{H} \bm{x}^{\perp} +(\bm{x}^{\parallel})^{\dag} \mathbf{H} \bm{x}^{\parallel} + 2 (\bm{x}^{\perp})^{\dag} \mathbf{H} \bm{x}^{\parallel} \\
        &= (\bm{x}^{\perp})^{\dag} \mathbf{H} \bm{x}^{\perp},
    \end{align*}
    since $\bm{x}^{\parallel} \in \ker \mathbf{F} \subseteq \ker \mathbf{H}$.
    If $\mathbf{X} = \bm{x} \bm{x}^{\dag}$ is an optimal solution to \eqref{equation:CCOPMomMatSDP}, then so is $\bm{x}^{\perp} (\bm{x}^{\perp})^{\dag}$, whereby $\bm{x}^{\perp} \in (\ker \mathbf{F})^{\perp}$.
\end{proof}

For the SDP \eqref{equation:CCOPSDP}, \autoref{proposition:SDPtoQCQP} suggests to re-scale the objective Hamiltonian such that $C \preceq \one$.
This can be achieved, e.g., by considering any upper-bound on the CCOP's objective function\footnote{
    Trivial upper bounds are readily available even for most NP problems.
    For example, for the Travelling Salesperson Problem, one may sum up the $n$ largest weights,
    and for the knapsack problem, one can sum up all item profits.
} to divide all objective values by.
To understand that this already suffices for \eqref{equation:CCOPMomMatSDP} to satisfy the prerequisites of \autoref{proposition:SDPtoQCQP} as well, we note that $T(\one) - T(C) \succeq 0$ if $\one - C \succeq 0$ since $T$ is linear.

\section{\label{section:NumericalExperiments}Numerical experiments}

As a proof of concept we apply the method to ten strongly correlated Pisinger knapsack instances~\cite{Pisinger2005WhereAreTheHardKnapsackProblems} with $n=16$ items, a class where ratio-greedy rounding is deliberately confounded.
Each item is mapped to one qubit, and the ansatz consists of $L = 16$ LCU channels, each carrying $\ell = 17$ operations: the identity and one Pauli-$Y$ on every qubit.
Starting from the uniform superposition $\ket{+}^{\otimes 16}$, the channel weights are initialised such that the initial output state is exactly the greedy solution and the optimisation is a strict improvement scheme over the classical baseline.
Channels are then updated one at a time by solving the $\ell$-dimensional constrained generalised eigenvalue problem~\eqref{equation:GeneralisedEigenvalueProblem}, sweeping back and forth over the channels for two cycles (61 steps).
The sample matrices are computed exactly using the simulation library Orkan~\cite{Ziegler2026_Orkan}.
\autoref{figure:NumericalExperiments} depicts the approximation ratio $\braket{\phi | C | \phi}/c_{\text{opt}}$ and the aggregated implementation probability of the full LCU sequence at every step.
The mean value never falls below the greedy value, and on the four instances with the weakest greedy solutions the optimiser leaves the classical point and converges to feasible superpositions (infeasible weight below $10^{-17}$) with approximation ratios of $0.98$, gaining up to $0.09$ over greedy; the remaining six instances are fixed points of the update.
These gains are obtained by coherent superpositions of item exchanges that no single-channel classical update can realise.
The implementation probability stays within a factor of six of its warm-start value $2^{-16}$ ($2^{-1}$ for each LCU).
However, note that initially all channels effectively implement single-qubit $R_Y(\pi/2)$ as an LCU operation which any advanced compiler could recognise and replace, thereby elevating the probability to one.
Updates to deterministic operations are marked by initial jumps in the implementation probability.
The subsequent slow decay is due to the spreading of weight over more constituents in the channels.

\section{\label{section:DiscussionAndConclusion}Discussion and Conclusions}

In this article, we have introduced a natural extension of recent proposals for non-unitary parameterised quantum circuits (PQCs) for unconstrained/soft-constrained to general hard-constrained combinatorial optimisation problems.
It aims at a unified framework for addressing arbitrarily-constrained combinatorial optimisation problems already on quantum devices with short to medium coherence times.
However, our approach also finds applicability on far-term devices which may still outsource suitable subroutines to CPUs.
Like its predecessors, our hybrid framework introduces an ansatz class reachable by (non-unitary) linear combinations of unitaries, parameterised by a complex coefficient vector, an ansatz that naturally mitigates the effect of barren plateaus and local minima.
Previous methods formulated a generalised eigenvalue problem (GEP) for the unconstrained case or, more generally, a semidefinite program.
Within this article, we extended the derivation of a GEP for the constrained case and also fill some gaps in prior derivations with mathematical rigour.
This problem transformation requires recording several expectation values and state overlaps, for which we introduce a specialised measurement protocol, solely based on sampling bit strings from the main register and applying the classical objective function and a classical feasibility oracle.
In fact, the objective Hamiltonian does not even have to be implemented on the quantum computer.

Our framework draws inspiration not only from preceding non-unitary methods, but also from the Quantum Alternating Operator Ansatz (QAOA) which first gave a formal treatment of ensuring hard-constraints with PQCs.
One of the QAOA's biggest strengths---its general formalism---is also its major drawback:
There is no universal design pattern for constructing feasibility-preserving and well-mixing PQCs.
With this work, we restore out-of-the-box applicability by integrating the classical objective function and feasibility oracle into the framework;
yet, our approach is universally applicable, sharing this utmost important feature with the QAOA.

With this article, we have laid the foundation for a unified framework and provide a proof of concept via a numerical simulation.
Further theoretical strengthening as well as extensive numerical experiments are crucial.
We will further gauge the performance of our framework on prominent examples such as the Travelling Salesperson Problem.

\IEEEtriggeratref{43}
\bibliographystyle{IEEEtran}
\bibliography{IEEEabrv,bibliography}

\begin{thebibliography}{10}
\providecommand{\url}[1]{#1}
\csname url@samestyle\endcsname
\providecommand{\newblock}{\relax}
\providecommand{\bibinfo}[2]{#2}
\providecommand{\BIBentrySTDinterwordspacing}{\spaceskip=0pt\relax}
\providecommand{\BIBentryALTinterwordstretchfactor}{4}
\providecommand{\BIBentryALTinterwordspacing}{\spaceskip=\fontdimen2\font plus
\BIBentryALTinterwordstretchfactor\fontdimen3\font minus
  \fontdimen4\font\relax}
\providecommand{\BIBforeignlanguage}[2]{{%
\expandafter\ifx\csname l@#1\endcsname\relax
\typeout{** WARNING: IEEEtran.bst: No hyphenation pattern has been}%
\typeout{** loaded for the language `#1'. Using the pattern for}%
\typeout{** the default language instead.}%
\else
\language=\csname l@#1\endcsname
\fi
#2}}
\providecommand{\BIBdecl}{\relax}
\BIBdecl

\bibitem{PsiquantumTeam2025AManufacturablePlatformForPhotonicQuantumComputing}
\BIBentryALTinterwordspacing
{PsiQuantum team}, ``{A manufacturable platform for photonic quantum
  computing},'' \emph{Nature}, vol. 641, no. 8064, pp. 876--883, 2025.
  [Online]. Available: \url{https://doi.org/10.1038/s41586-025-08820-7}
\BIBentrySTDinterwordspacing

\bibitem{GoogleQuantumAi2023SuppressingQuantumErrorsByScalingASurfaceCodeLogicalQubit}
\BIBentryALTinterwordspacing
{Google Quantum AI}, ``{Suppressing quantum errors by scaling a surface code
  logical qubit},'' \emph{Nature}, vol. 614, no. 7949, pp. 676--681, 2023.
  [Online]. Available: \url{https://doi.org/10.1038/s41586-022-05434-1}
\BIBentrySTDinterwordspacing

\bibitem{Preskill2018QuantumComputingInTheNisqEraAndBeyond}
\BIBentryALTinterwordspacing
J.~Preskill, ``{Quantum Computing in the NISQ era and beyond},''
  \emph{Quantum}, vol.~2, p.~79, 2018. [Online]. Available:
  \url{https://doi.org/10.22331/q-2018-08-06-79}
\BIBentrySTDinterwordspacing

\bibitem{Grover1997QuantumMechanicsHelpsInSearchingForANeedleInAHaystack}
\BIBentryALTinterwordspacing
L.~K. Grover, ``{Quantum Mechanics Helps in Searching for a Needle in a
  Haystack},'' \emph{Phys. Rev. Lett.}, vol.~79, no.~2, pp. 325--328, 1997.
  [Online]. Available: \url{https://doi.org/10.1103/physrevlett.79.325}
\BIBentrySTDinterwordspacing

\bibitem{Kadowaki1998QuantumAnnealingInTheTransverseIsingModel}
\BIBentryALTinterwordspacing
T.~Kadowaki and H.~Nishimori, ``{Quantum annealing in the transverse Ising
  model},'' \emph{Phys. Rev.}, vol.~58, no.~5, pp. 5355--5363, 1998. [Online].
  Available: \url{https://doi.org/10.1103/physreve.58.5355}
\BIBentrySTDinterwordspacing

\bibitem{Farhi2000QuantumComputationByAdiabaticEvolution}
E.~Farhi, J.~Goldstone, S.~Gutmann, and M.~Sipser, ``{Quantum Computation by
  Adiabatic Evolution},'' 2000, [arXiv preprint
  \href{https://arxiv.org/abs/quant-ph/0001106}{arXiv:quant-ph/0001106}].

\bibitem{Abbas2024ChallengesAndOpportunitiesInQuantumOptimization}
\BIBentryALTinterwordspacing
A.~Abbas, A.~Ambainis, B.~Augustino, A.~B{\"a}rtschi, H.~Buhrman \emph{et~al.},
  ``{Challenges and opportunities in quantum optimization},'' \emph{Nat. Rev.
  Phys.}, vol.~6, no.~12, pp. 718--735, 2024. [Online]. Available:
  \url{https://doi.org/10.1038/s42254-024-00770-9}
\BIBentrySTDinterwordspacing

\bibitem{Montanaro2020QuantumSpeedupOfBranchAndBoundAlgorithms}
\BIBentryALTinterwordspacing
A.~Montanaro, ``{Quantum speedup of branch-and-bound algorithms},'' \emph{Phys.
  Rev. Res.}, vol.~2, no.~1, p. 013056, 2020. [Online]. Available:
  \url{https://doi.org/10.1103/physrevresearch.2.013056}
\BIBentrySTDinterwordspacing

\bibitem{Ambainis2019QuantumSpeedupsForExponentialTimeDynamicProgrammingAlgorithms}
\BIBentryALTinterwordspacing
A.~Ambainis, K.~Balodis, J.~Iraids, M.~Kokainis, K.~Pr{\=u}sis \emph{et~al.},
  ``{Quantum Speedups for Exponential-Time Dynamic Programming Algorithms},''
  in \emph{Proceedings of the Thirtieth Annual ACM-SIAM Symposium on Discrete
  Algorithms}.\hskip 1em plus 0.5em minus 0.4em\relax Society for Industrial
  and Applied Mathematics, 2019, pp. 1783--1793. [Online]. Available:
  \url{https://doi.org/10.1137/1.9781611975482.107}
\BIBentrySTDinterwordspacing

\bibitem{VanApeldoorn2020QuantumSdpSolversBetterUpperAndLowerBounds}
\BIBentryALTinterwordspacing
J.~van Apeldoorn, A.~Gily{\'e}n, S.~Gribling, and R.~de~Wolf, ``{Quantum
  SDP-Solvers: Better upper and lower bounds},'' \emph{Quantum}, vol.~4, p.
  230, 2020. [Online]. Available:
  \url{https://doi.org/10.22331/q-2020-02-14-230}
\BIBentrySTDinterwordspacing

\bibitem{Cerezo2021VariationalQuantumAlgorithms}
\BIBentryALTinterwordspacing
M.~Cerezo, A.~Arrasmith, R.~Babbush, S.~C. Benjamin, S.~Endo \emph{et~al.},
  ``{Variational quantum algorithms},'' \emph{Nat. Rev. Phys.}, vol.~3, no.~9,
  pp. 625--644, 2021. [Online]. Available:
  \url{https://doi.org/10.1038/s42254-021-00348-9}
\BIBentrySTDinterwordspacing

\bibitem{Hadfield2019FromTheQuantumApproximateOptimizationAlgorithmToAQuantumAlternatingOperatorAnsatz}
\BIBentryALTinterwordspacing
S.~Hadfield, Z.~Wang, B.~O’Gorman, E.~G. Rieffel, D.~Venturelli
  \emph{et~al.}, ``{From the Quantum Approximate Optimization Algorithm to a
  Quantum Alternating Operator Ansatz},'' \emph{Algorithms}, vol.~12, no.~2,
  p.~34, 2019. [Online]. Available: \url{https://doi.org/10.3390/a12020034}
\BIBentrySTDinterwordspacing

\bibitem{Bittel2021TrainingVariationalQuantumAlgorithmsIsNpHard}
\BIBentryALTinterwordspacing
L.~Bittel and M.~Kliesch, ``{Training Variational Quantum Algorithms Is
  NP-Hard},'' \emph{Phys. Rev. Lett.}, vol. 127, no.~12, p. 120502, 2021.
  [Online]. Available: \url{https://doi.org/10.1103/physrevlett.127.120502}
\BIBentrySTDinterwordspacing

\bibitem{Larocca2025BarrenPlateausInVariationalQuantumComputing}
\BIBentryALTinterwordspacing
M.~Larocca, S.~Thanasilp, S.~Wang, K.~Sharma, J.~Biamonte \emph{et~al.},
  ``{Barren plateaus in variational quantum computing},'' \emph{Nat. Rev.
  Phys.}, vol.~7, no.~4, pp. 174--189, 2025. [Online]. Available:
  \url{https://doi.org/10.1038/s42254-025-00813-9}
\BIBentrySTDinterwordspacing

\bibitem{Binkowski2025FromBarrenPlateausThroughFertileValleysConicExtensionsOfParameterisedQuantumCircuits}
\BIBentryALTinterwordspacing
L.~Binkowski, G.~Ko{\ss}mann, T.~J. Osborne, R.~Schwonnek, and T.~Ziegler,
  ``{From Barren Plateaus Through Fertile Valleys: Conic Extensions of
  Parameterised Quantum Circuits},'' in \emph{2025 IEEE International
  Conference on Quantum Computing and Engineering (QCE)}, 2025, pp. 111--118.
  [Online]. Available: \url{https://doi.org/10.1109/qce65121.2025.00022}
\BIBentrySTDinterwordspacing

\bibitem{Mcclean2017HybridQuantumClassicalHierarchyForMitigationOfDecoherenceAndDeterminationOfExcitedStates}
\BIBentryALTinterwordspacing
J.~R. McClean, M.~E. Kimchi-Schwartz, J.~Carter, and W.~A. de~Jong, ``{Hybrid
  quantum-classical hierarchy for mitigation of decoherence and determination
  of excited states},'' \emph{Phys. Rev.}, vol.~95, no.~4, p. 042308, 2017.
  [Online]. Available: \url{https://doi.org/10.1103/physreva.95.042308}
\BIBentrySTDinterwordspacing

\bibitem{Huggins2020ANonOrthogonalVariationalQuantumEigensolver}
\BIBentryALTinterwordspacing
W.~J. Huggins, J.~Lee, U.~Baek, B.~O’Gorman, and K.~B. Whaley, ``{A
  non-orthogonal variational quantum eigensolver},'' \emph{New J. Phys.},
  vol.~22, no.~7, p. 073009, 2020. [Online]. Available:
  \url{https://doi.org/10.1088/1367-2630/ab867b}
\BIBentrySTDinterwordspacing

\bibitem{Bharti2022NoisyIntermediateScaleQuantumAlgorithmForSemidefiniteProgramming}
\BIBentryALTinterwordspacing
K.~Bharti, T.~Haug, V.~Vedral, and L.~Kwek, ``{Noisy intermediate-scale quantum
  algorithm for semidefinite programming},'' \emph{Phys. Rev.}, vol. 105,
  no.~5, p. 052445, 2022. [Online]. Available:
  \url{https://doi.org/10.1103/physreva.105.052445}
\BIBentrySTDinterwordspacing

\bibitem{Childs2012HamiltonianSimulationUsingLinearCombinationsOfUnitaryOperations}
\BIBentryALTinterwordspacing
A.~M. Childs and N.~Wiebe, ``{Hamiltonian simulation using linear combinations
  of unitary operations},'' \emph{Quantum Inf. Comput.}, vol.~12, no. 11\&12,
  pp. 901--924, 2012. [Online]. Available:
  \url{https://doi.org/10.26421/qic12.11-12-1}
\BIBentrySTDinterwordspacing

\bibitem{Farhi2014AQuantumApproximateOptimizationAlgorithm}
E.~Farhi, J.~Goldstone, and S.~Gutmann, ``{A Quantum Approximate Optimization
  Algorithm},'' 2014, [arXiv preprint
  \href{https://arxiv.org/abs/1411.4028}{arXiv:1411.4028}].

\bibitem{Zhang2017ReviewOfJobShopSchedulingResearchAndItsNewPerspectivesUnderIndustry40}
\BIBentryALTinterwordspacing
J.~Zhang, G.~Ding, Y.~Zou, S.~Qin, and J.~Fu, ``{Review of job shop scheduling
  research and its new perspectives under Industry 4.0},'' \emph{J. Intell.
  Manuf.}, vol.~30, no.~4, pp. 1809--1830, 2017. [Online]. Available:
  \url{https://doi.org/10.1007/s10845-017-1350-2}
\BIBentrySTDinterwordspacing

\bibitem{Kellerer2004KnapsackProblems}
\BIBentryALTinterwordspacing
H.~Kellerer, U.~Pferschy, and D.~Pisinger, \emph{{Knapsack Problems}}.\hskip
  1em plus 0.5em minus 0.4em\relax Springer, 2004. [Online]. Available:
  \url{https://doi.org/10.1007/978-3-540-24777-7}
\BIBentrySTDinterwordspacing

\bibitem{Tarjan1977FindingAMaximumIndependentSet}
\BIBentryALTinterwordspacing
R.~E. Tarjan and A.~E. Trojanowski, ``{Finding a Maximum Independent Set},''
  \emph{SIAM J. Comput.}, vol.~6, no.~3, pp. 537--546, 1977. [Online].
  Available: \url{https://doi.org/10.1137/0206038}
\BIBentrySTDinterwordspacing

\bibitem{Applegate2009CertificationOfAnOptimalTspTourThrough85900Cities}
\BIBentryALTinterwordspacing
D.~L. Applegate, R.~E. Bixby, V.~Chv{\'a}tal, W.~Cook, D.~G. Espinoza
  \emph{et~al.}, ``{Certification of an optimal TSP tour through 85,900
  cities},'' \emph{Oper. Res. Lett.}, vol.~37, no.~1, pp. 11--15, 2009.
  [Online]. Available: \url{https://doi.org/10.1016/j.orl.2008.09.006}
\BIBentrySTDinterwordspacing

\bibitem{Karp1972ReducibilityAmongCombinatorialProblems}
\BIBentryALTinterwordspacing
R.~M. Karp, ``{Reducibility among Combinatorial Problems},'' in
  \emph{Complexity of Computer Computations}.\hskip 1em plus 0.5em minus
  0.4em\relax Boston, {MA}: Springer {US}, 1972, pp. 85--103. [Online].
  Available: \url{https://doi.org/10.1007/978-1-4684-2001-2\_9}
\BIBentrySTDinterwordspacing

\bibitem{Brandt1999QubitDevicesAndTheIssueOfQuantumDecoherence}
\BIBentryALTinterwordspacing
H.~E. Brandt, ``{Qubit devices and the issue of quantum decoherence},''
  \emph{Prog. Quantum Electronics}, vol.~22, no. 5-6, pp. 257--370, 1999.
  [Online]. Available: \url{https://doi.org/10.1016/s0079-6727(99)00003-8}
\BIBentrySTDinterwordspacing

\bibitem{Kuebler2020AnAdaptiveOptimizerForMeasurementFrugalVariationalAlgorithms}
\BIBentryALTinterwordspacing
J.~M. K{\"u}bler, A.~Arrasmith, L.~Cincio, and P.~J. Coles, ``{An Adaptive
  Optimizer for Measurement-Frugal Variational Algorithms},'' \emph{Quantum},
  vol.~4, p. 263, 2020. [Online]. Available:
  \url{https://doi.org/10.22331/q-2020-05-11-263}
\BIBentrySTDinterwordspacing

\bibitem{BonetMonroig2023PerformanceComparisonOfOptimizationMethodsOnVariationalQuantumAlgorithms}
\BIBentryALTinterwordspacing
X.~Bonet-Monroig, H.~Wang, D.~Vermetten, B.~Senjean, C.~Moussa \emph{et~al.},
  ``{Performance comparison of optimization methods on variational quantum
  algorithms},'' \emph{Phys. Rev.}, vol. 107, no.~3, p. 032407, 2023. [Online].
  Available: \url{https://doi.org/10.1103/physreva.107.032407}
\BIBentrySTDinterwordspacing

\bibitem{Binkowski2024ElementaryProofOfQaoaConvergence}
\BIBentryALTinterwordspacing
L.~Binkowski, G.~Ko{\ss}mann, T.~Ziegler, and R.~Schwonnek, ``{Elementary proof
  of QAOA convergence},'' \emph{New J. Phys.}, vol.~26, no.~7, p. 073001, 2024.
  [Online]. Available: \url{https://doi.org/10.1088/1367-2630/ad59bb}
\BIBentrySTDinterwordspacing

\bibitem{GrandRive2019KnapsackProblemVariantsOfQaoaForBatteryRevenueOptimisation}
P.~D. d.~l. Grand'rive and J.~Hullo, ``{Knapsack Problem variants of QAOA for
  battery revenue optimisation},'' 2019, [arXiv preprint
  \href{https://arxiv.org/abs/1908.02210}{arXiv:1908.02210}].

\bibitem{Bartschi2020GroverMixersForQaoaShiftingComplexityFromMixerDesignToStatePreparation}
\BIBentryALTinterwordspacing
A.~Bartschi and S.~Eidenbenz, ``{Grover Mixers for QAOA: Shifting Complexity
  from Mixer Design to State Preparation},'' in \emph{2020 IEEE International
  Conference on Quantum Computing and Engineering (QCE)}, 2020, pp. 72--82.
  [Online]. Available: \url{https://doi.org/10.1109/qce49297.2020.00020}
\BIBentrySTDinterwordspacing

\bibitem{Mcclean2018BarrenPlateausInQuantumNeuralNetworkTrainingLandscapes}
\BIBentryALTinterwordspacing
J.~R. McClean, S.~Boixo, V.~N. Smelyanskiy, R.~Babbush, and H.~Neven, ``{Barren
  plateaus in quantum neural network training landscapes},'' \emph{Nat.
  Commun.}, vol.~9, no.~1, p. 4812, 2018. [Online]. Available:
  \url{https://doi.org/10.1038/s41467-018-07090-4}
\BIBentrySTDinterwordspacing

\bibitem{Arrasmith2021EffectOfBarrenPlateausOnGradientFreeOptimization}
\BIBentryALTinterwordspacing
A.~Arrasmith, M.~Cerezo, P.~Czarnik, L.~Cincio, and P.~J. Coles, ``{Effect of
  barren plateaus on gradient-free optimization},'' \emph{Quantum}, vol.~5, p.
  558, 2021. [Online]. Available:
  \url{https://doi.org/10.22331/q-2021-10-05-558}
\BIBentrySTDinterwordspacing

\bibitem{Wang2021NoiseInducedBarrenPlateausInVariationalQuantumAlgorithms}
\BIBentryALTinterwordspacing
S.~Wang, E.~Fontana, M.~Cerezo, K.~Sharma, A.~Sone \emph{et~al.},
  ``{Noise-induced barren plateaus in variational quantum algorithms},''
  \emph{Nat. Commun.}, vol.~12, no.~1, p. 6961, 2021. [Online]. Available:
  \url{https://doi.org/10.1038/s41467-021-27045-6}
\BIBentrySTDinterwordspacing

\bibitem{OrtizMarrero2021EntanglementInducedBarrenPlateaus}
\BIBentryALTinterwordspacing
C.~Ortiz~Marrero, M.~Kieferov{\'a}, and N.~Wiebe, ``{Entanglement-Induced
  Barren Plateaus},'' \emph{PRX Quantum}, vol.~2, no.~4, p. 040316, 2021.
  [Online]. Available: \url{https://doi.org/10.1103/prxquantum.2.040316}
\BIBentrySTDinterwordspacing

\bibitem{Cerezo2021CostFunctionDependentBarrenPlateausInShallowParametrizedQuantumCircuits}
\BIBentryALTinterwordspacing
M.~Cerezo, A.~Sone, T.~Volkoff, L.~Cincio, and P.~J. Coles, ``{Cost function
  dependent barren plateaus in shallow parametrized quantum circuits},''
  \emph{Nat. Commun.}, vol.~12, no.~1, p. 1791, 2021. [Online]. Available:
  \url{https://doi.org/10.1038/s41467-021-21728-w}
\BIBentrySTDinterwordspacing

\bibitem{Ekert2002DirectEstimationsOfLinearAndNonlinearFunctionalsOfAQuantumState}
\BIBentryALTinterwordspacing
A.~K. Ekert, C.~M. Alves, D.~K.~L. Oi, M.~Horodecki, P.~Horodecki
  \emph{et~al.}, ``{Direct Estimations of Linear and Nonlinear Functionals of a
  Quantum State},'' \emph{Phys. Rev. Lett.}, vol.~88, no.~21, p. 217901, 2002.
  [Online]. Available: \url{https://doi.org/10.1103/physrevlett.88.217901}
\BIBentrySTDinterwordspacing

\bibitem{Bharti2021IterativeQuantumAssistedEigensolver}
\BIBentryALTinterwordspacing
K.~Bharti and T.~Haug, ``{Iterative quantum-assisted eigensolver},''
  \emph{Phys. Rev.}, vol. 104, no.~5, 2021. [Online]. Available:
  \url{https://doi.org/10.1103/physreva.104.l050401}
\BIBentrySTDinterwordspacing

\bibitem{Thoai2000DualityBoundMethodForTheGeneralQuadraticProgrammingProblemWithQuadraticConstraints}
\BIBentryALTinterwordspacing
N.~V. Thoai, ``{Duality Bound Method for the General Quadratic Programming
  Problem with Quadratic Constraints},'' \emph{J. Optim. Theory Appl.}, vol.
  107, no.~2, pp. 331--354, 2000. [Online]. Available:
  \url{https://doi.org/10.1023/a:1026437621223}
\BIBentrySTDinterwordspacing

\bibitem{Nguyen2022StrongDualityForGeneralQuadraticProgramsWithQuadraticEqualityConstraints}
\BIBentryALTinterwordspacing
D.~Nguyen, ``{Strong Duality for General Quadratic Programs with Quadratic
  Equality Constraints},'' \emph{J. Optim. Theory Appl.}, vol. 195, no.~1, pp.
  297--313, 2022. [Online]. Available:
  \url{https://doi.org/10.1007/s10957-022-02082-3}
\BIBentrySTDinterwordspacing

\bibitem{Kato1995PerturbationTheoryForLinearOperators}
\BIBentryALTinterwordspacing
T.~Kato, \emph{{Perturbation Theory for Linear Operators}}.\hskip 1em plus
  0.5em minus 0.4em\relax Berlin, Heidelberg: Springer, 1995. [Online].
  Available: \url{https://doi.org/10.1007/978-3-642-66282-9}
\BIBentrySTDinterwordspacing

\bibitem{Ramacciotti2024SimpleQuantumAlgorithmToEfficientlyPrepareSparseStates}
\BIBentryALTinterwordspacing
D.~Ramacciotti, A.~I. Lefterovici, and A.~F. Rotundo, ``{Simple quantum
  algorithm to efficiently prepare sparse states},'' \emph{Phys. Rev.}, vol.
  110, no.~3, p. 032609, 2024. [Online]. Available:
  \url{https://doi.org/10.1103/physreva.110.032609}
\BIBentrySTDinterwordspacing

\bibitem{Binkowski2025SymmetryBasedQuantumAlgorithmsForOpenShopSchedulingWithHardConstraints}
\BIBentryALTinterwordspacing
L.~Binkowski, G.~Ko{\ss}mann, C.~Tutschku, and R.~Schwonnek, ``{Symmetry-based
  quantum algorithms for open-shop scheduling with hard constraints},''
  \emph{Acad. Quantum}, vol.~2, no.~3, 2025. [Online]. Available:
  \url{https://doi.org/10.20935/acadquant7900}
\BIBentrySTDinterwordspacing

\bibitem{Schwiering2026_ExhaustiveAndFeasibleParametrisationWithApplicationsToTheTravellingSalespersonProblem}
\BIBentryALTinterwordspacing
M.~Schwiering, T.~Ziegler, L.~Binkowski, and B.~Sambale, ``Exhaustive and
  feasible parametrisation with applications to the travelling salesperson
  problem,'' 2026. [Online]. Available: \url{https://arxiv.org/abs/2604.24297}
\BIBentrySTDinterwordspacing

\bibitem{Huang2007ComplexMatrixDecompositionAndQuadraticProgramming}
\BIBentryALTinterwordspacing
Y.~Huang and S.~Zhang, ``{Complex Matrix Decomposition and Quadratic
  Programming},'' \emph{Math. Oper. Res.}, vol.~32, no.~3, pp. 758--768, 2007.
  [Online]. Available: \url{https://doi.org/10.1287/moor.1070.0268}
\BIBentrySTDinterwordspacing

\bibitem{Pisinger2005WhereAreTheHardKnapsackProblems}
\BIBentryALTinterwordspacing
D.~Pisinger, ``{Where are the hard knapsack problems?}'' \emph{Comput. \& Oper.
  Res.}, vol.~32, no.~9, pp. 2271--2284, 2005. [Online]. Available:
  \url{https://doi.org/10.1016/j.cor.2004.03.002}
\BIBentrySTDinterwordspacing

\bibitem{Ziegler2026_Orkan}
\BIBentryALTinterwordspacing
T.~Ziegler, ``Orkan: Cache-friendly simulation of quantum operations on
  hermitian operators,'' 2026. [Online]. Available:
  \url{https://arxiv.org/abs/2604.15765}
\BIBentrySTDinterwordspacing

\end{thebibliography}

\end{document}